%% file: clifford+Tcompleteness-arxiv.tex
\documentclass{article}
\usepackage{amssymb,nicefrac}
\usepackage{amsfonts}
\usepackage{graphicx}
\usepackage[fleqn]{amsmath}
\usepackage[varg]{txfonts}
\usepackage{stmaryrd}

\usepackage{framed}
\usepackage{color}
\usepackage{ulem}
\usepackage{tikzfig}
\input{tikzstyles.tex}

\input{defs.tex}

\usepackage{mathtools}
\usepackage{placeins}
\newtheorem{Th}{Theorem}[section]
\newtheorem{theorem}[Th]{Theorem}
\newtheorem{proposition}[Th]{Proposition}
\newtheorem{lemma}[Th]{Lemma}

\newenvironment{proof}{\textbf{Proof:}}{\hfill$\Box$\newline}

\title{ Completeness of the ZX-calculus for Pure Qubit Clifford+T Quantum Mechanics}
\author{Kang Feng Ng \qquad\qquad Quanlong Wang\\ Department of Computer Science, University of Oxford }

\begin{document}

\date{}\maketitle

\begin{abstract}
Recently, we gave a complete axiomatisation of the ZX-calculus for the overall pure qubit quantum mechanics. Based on this result, here we also obtain a complete axiomatisation of the ZX-calculus for the Clifford+T quantum mechanics
by restricting the ring of complex numbers to its subring corresponding to the Clifford+T fragment resting on the completeness theorem of the ZW-calculus for arbitrary commutative ring. In contrast to the first complete axiomatisation of the ZX-calculus for the Clifford+T fragment, we have two new generators as features rather than novelties: the triangle can be employed as an essential component to construct a Toffoli gate in a very simple form, while the  $\lambda$ box can be slightly extended to a generalised phase so that the generalised supplementarity (cyclotomic supplementarity ) is naturally seen as a special case of the generalised spider rule.


\end{abstract}

\section{Introduction}
Clifford+T qubit quantum mechanics (QM) is an approximatively universal fragment of QM, which has been widely used in quantum computing. In contrast to the traditional way of matrix calculation,  the ZX-calculus introduced by Coecke and Duncan \cite{CoeckeDuncan}  is a diagrammatical axiomatisation of quantum computing.  One of the main open problems of the ZX-calculus  is to give a complete axiomatisation for the Clifford+T QM \cite{cqmwiki}.
After the first completeness result on this fragment
--the completeness of the ZX-calculus for single qubit Clifford+T QM \cite{Miriam1ct}, there finally comes an completion of the ZX-Calculus for the whole Clifford+T QM  \cite{Emmanuel}, 
which contributes a solution to the above mentioned open problem.

Further to the complete axiomatisation of the ZX-Calculus for the Clifford+T fragment QM, we have given a complete axiomatisation of the ZX-calculus for the overall pure qubit QM \cite{ngwang}. In this paper, we first simplify the rule of addition (AD)  and show that some rules can be derived from other rules in \cite{ngwang}. Then we obtain a complete axiomatisation of the ZX-calculus for the Clifford+T quantum mechanics
by restricting the ring of complex numbers to its subring $\mathbb{Z}[i, \frac{1}{\sqrt{2}}]$ based on the completeness theorem of the ZW-calculus for arbitrary commutative ring \cite{amar}.  In comparison to the completeness proof in  \cite{ngwang},  a modification of the interpretation from the ZW-calculus to the ZX-calculus is made here. 

The main difference between the two complete axiomatisations of the ZX-Calculus for the Clifford+T fragment QM shown in this paper and that presented in \cite{Emmanuel} is as follows:
\begin{enumerate}
\item  Although the number of rules (which is 30) listed here is much more than that of \cite{Emmanuel} (which is 13), the number of nodes in each non-scalar diagram of the extended part of rules (non-stablizer part) is at most 8 in this paper, in contrast to a maximum of 17 in \cite{Emmanuel}. 
 
\item Following \cite{ngwang}, we have still introduced two more generators-- the triangle and the $\lambda$ box-- in this paper, while there are only green nodes and red nodes as generators  in \cite{Emmanuel}. Our new generators are features rather than novelties: the triangle can be employed as an essential component to construct a Toffoli gate in a very simple form, while the  $\lambda$ box can be slightly extended to a generalised phase so that the generalised supplementarity (also called cyclotomic supplementarity, with supplementarity as a special case) \cite{jpvw} is naturally seen as a special case of the generalised spider rule. These features are explained in detail in section \ref{zxforct}. 

\item The translation from the ZX-calculus to the ZW-calculus in our paper is more direct . 
 \end{enumerate}

\section{Universal completion of the ZX-calculus}
In the following we list the rules of the ZX-calculus for the overall pure qubit QM as shown in \cite{ngwang}:

\begin{figure}[!h]
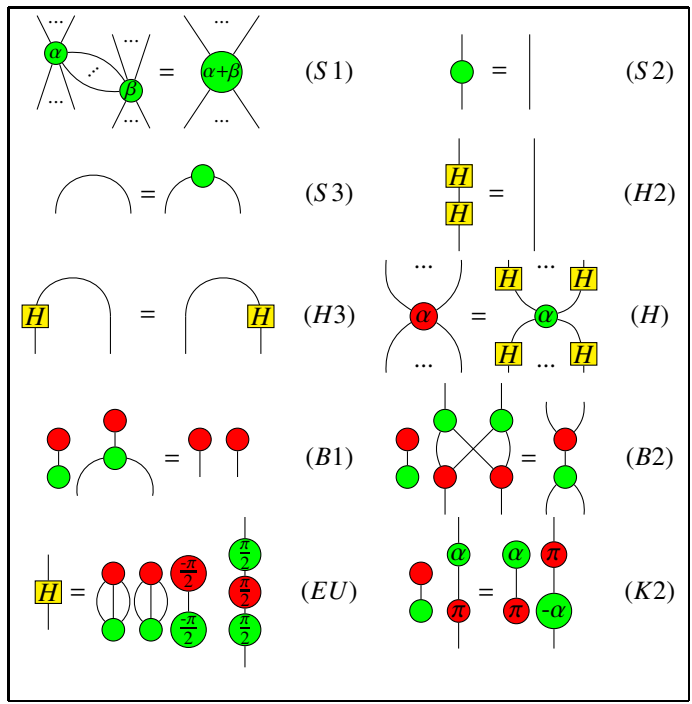

\begin{center}
\[
\quad \qquad\begin{array}{|cccc|}
\hline
\beginpgfgraphicnamed{diagrams//spiderlt}
\InputIfFileExists{diagrams//spiderlt.tikz}{}{\input{./figures/diagrams//spiderlt.tikz}}
\endpgfgraphicnamed=%
\beginpgfgraphicnamed{diagrams//spiderrt}
\InputIfFileExists{diagrams//spiderrt.tikz}{}{\input{./figures/diagrams//spiderrt.tikz}}
\endpgfgraphicnamed &(S1) &%
\beginpgfgraphicnamed{diagrams//s2new}
\InputIfFileExists{diagrams//s2new.tikz}{}{\input{./figures/diagrams//s2new.tikz}}
\endpgfgraphicnamed &(S2)\\
\beginpgfgraphicnamed{diagrams//induced_compact_structure-2wirelt}
\InputIfFileExists{diagrams//induced_compact_structure-2wirelt.tikz}{}{\input{./figures/diagrams//induced_compact_structure-2wirelt.tikz}}
\endpgfgraphicnamed=%
\beginpgfgraphicnamed{diagrams//induced_compact_structure-2wirert}
\InputIfFileExists{diagrams//induced_compact_structure-2wirert.tikz}{}{\input{./figures/diagrams//induced_compact_structure-2wirert.tikz}}
\endpgfgraphicnamed&(S3) & %
\beginpgfgraphicnamed{diagrams//hsquare}
\InputIfFileExists{diagrams//hsquare.tikz}{}{\input{./figures/diagrams//hsquare.tikz}}
\endpgfgraphicnamed &(H2)\\
\beginpgfgraphicnamed{diagrams//hslidecap}
\InputIfFileExists{diagrams//hslidecap.tikz}{}{\input{./figures/diagrams//hslidecap.tikz}}
\endpgfgraphicnamed &(H3) &%
\beginpgfgraphicnamed{diagrams//h2newlt}
\InputIfFileExists{diagrams//h2newlt.tikz}{}{\input{./figures/diagrams//h2newlt.tikz}}
\endpgfgraphicnamed=%
\beginpgfgraphicnamed{diagrams//h2newrt}
\InputIfFileExists{diagrams//h2newrt.tikz}{}{\input{./figures/diagrams//h2newrt.tikz}}
\endpgfgraphicnamed&(H)\\
\beginpgfgraphicnamed{diagrams//b1slt}
\InputIfFileExists{diagrams//b1slt.tikz}{}{\input{./figures/diagrams//b1slt.tikz}}
\endpgfgraphicnamed=%
\beginpgfgraphicnamed{diagrams//b1srt}
\InputIfFileExists{diagrams//b1srt.tikz}{}{\input{./figures/diagrams//b1srt.tikz}}
\endpgfgraphicnamed&(B1) & %
\beginpgfgraphicnamed{diagrams//b2slt}
\InputIfFileExists{diagrams//b2slt.tikz}{}{\input{./figures/diagrams//b2slt.tikz}}
\endpgfgraphicnamed=%
\beginpgfgraphicnamed{diagrams//b2srt}
\InputIfFileExists{diagrams//b2srt.tikz}{}{\input{./figures/diagrams//b2srt.tikz}}
\endpgfgraphicnamed&(B2)\\
\beginpgfgraphicnamed{diagrams//HadaDecomSingleslt}
\InputIfFileExists{diagrams//HadaDecomSingleslt.tikz}{}{\input{./figures/diagrams//HadaDecomSingleslt.tikz}}
\endpgfgraphicnamed= %
\beginpgfgraphicnamed{diagrams//HadaDecomSinglesrt}
\InputIfFileExists{diagrams//HadaDecomSinglesrt.tikz}{}{\input{./figures/diagrams//HadaDecomSinglesrt.tikz}}
\endpgfgraphicnamed&(EU)    & %
\beginpgfgraphicnamed{diagrams//k2slt}
\InputIfFileExists{diagrams//k2slt.tikz}{}{\input{./figures/diagrams//k2slt.tikz}}
\endpgfgraphicnamed=%
\beginpgfgraphicnamed{diagrams//k2srt}
\InputIfFileExists{diagrams//k2srt.tikz}{}{\input{./figures/diagrams//k2srt.tikz}}
\endpgfgraphicnamed&(K2)\\

&&&\\ 
\hline
\end{array}\]
\end{center}
  \caption{Standard ZX-calculus rules, where $\alpha, \beta\in [0,~2\pi)$.}\label{figure1}  
  \end{figure}
  
  \begin{figure}[!h]
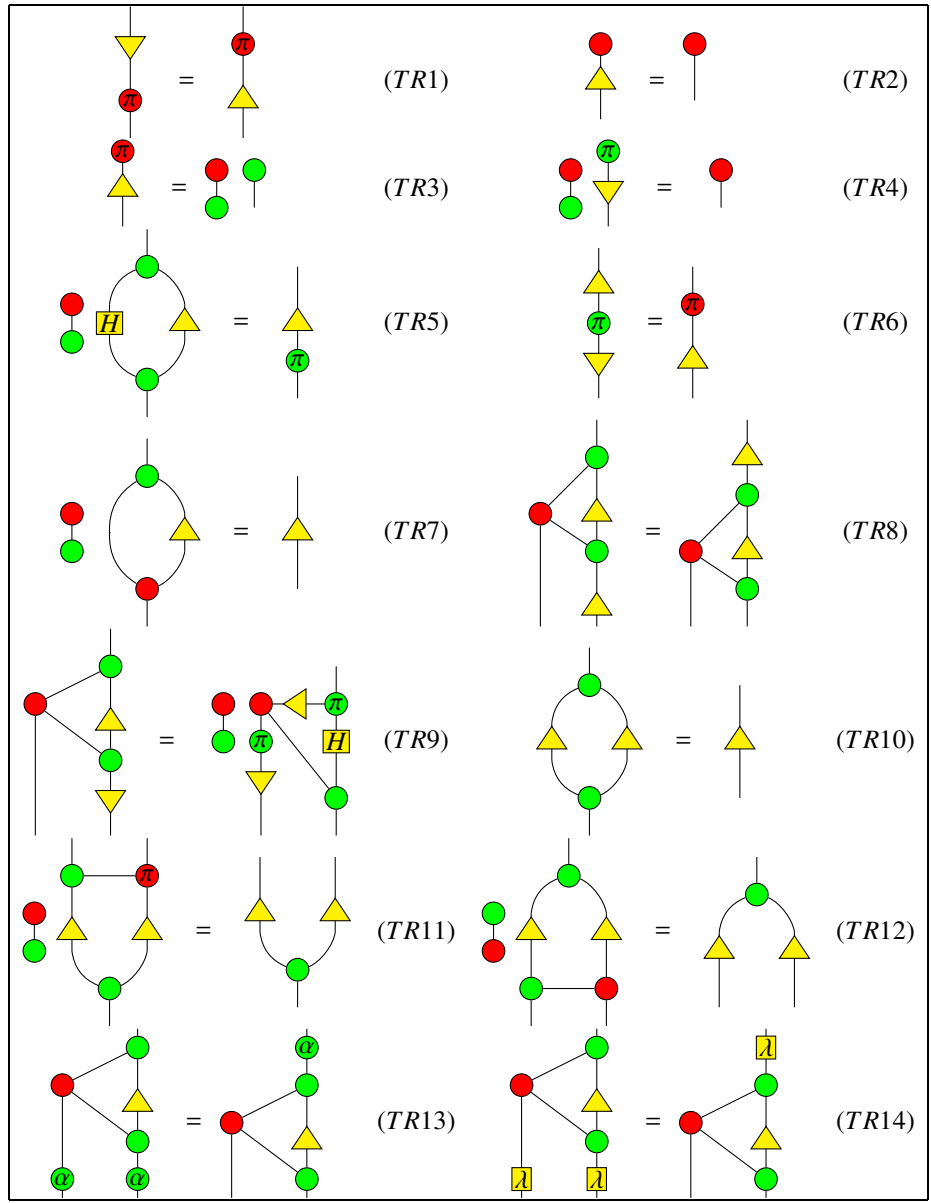

  	\begin{center}
  		\[
  		\quad \qquad\begin{array}{|cccc|}
  		\hline
  		%
\beginpgfgraphicnamed{diagrams//trianglepicommute}
\InputIfFileExists{diagrams//trianglepicommute.tikz}{}{\input{./figures/diagrams//trianglepicommute.tikz}}
\endpgfgraphicnamed &(TR1) &%
\beginpgfgraphicnamed{diagrams//triangleocopy}
\InputIfFileExists{diagrams//triangleocopy.tikz}{}{\input{./figures/diagrams//triangleocopy.tikz}}
\endpgfgraphicnamed &(TR2)\\

\beginpgfgraphicnamed{diagrams//trianglepicopy}
\InputIfFileExists{diagrams//trianglepicopy.tikz}{}{\input{./figures/diagrams//trianglepicopy.tikz}}
\endpgfgraphicnamed&(TR3) & %
\beginpgfgraphicnamed{diagrams//trianglegdpicopy}
\InputIfFileExists{diagrams//trianglegdpicopy.tikz}{}{\input{./figures/diagrams//trianglegdpicopy.tikz}}
\endpgfgraphicnamed &(TR4)\\

\beginpgfgraphicnamed{diagrams//trianglehhopf}
\InputIfFileExists{diagrams//trianglehhopf.tikz}{}{\input{./figures/diagrams//trianglehhopf.tikz}}
\endpgfgraphicnamed &(TR5) &%
\beginpgfgraphicnamed{diagrams//gpiintriangles}
\InputIfFileExists{diagrams//gpiintriangles.tikz}{}{\input{./figures/diagrams//gpiintriangles.tikz}}
\endpgfgraphicnamed&(TR6)\\

\beginpgfgraphicnamed{diagrams//trianglehopf}
\InputIfFileExists{diagrams//trianglehopf.tikz}{}{\input{./figures/diagrams//trianglehopf.tikz}}
\endpgfgraphicnamed&(TR7) & %
\beginpgfgraphicnamed{diagrams//2triangleup}
\InputIfFileExists{diagrams//2triangleup.tikz}{}{\input{./figures/diagrams//2triangleup.tikz}}
\endpgfgraphicnamed&(TR8)\\

\beginpgfgraphicnamed{diagrams//2triangledown}
\InputIfFileExists{diagrams//2triangledown.tikz}{}{\input{./figures/diagrams//2triangledown.tikz}}
\endpgfgraphicnamed&(TR9)    & %
\beginpgfgraphicnamed{diagrams//2trianglehopf}
\InputIfFileExists{diagrams//2trianglehopf.tikz}{}{\input{./figures/diagrams//2trianglehopf.tikz}}
\endpgfgraphicnamed&(TR10)\\

\beginpgfgraphicnamed{diagrams//2triangledeloop}
\InputIfFileExists{diagrams//2triangledeloop.tikz}{}{\input{./figures/diagrams//2triangledeloop.tikz}}
\endpgfgraphicnamed &(TR11) &%
\beginpgfgraphicnamed{diagrams//2triangledeloopnopi}
\InputIfFileExists{diagrams//2triangledeloopnopi.tikz}{}{\input{./figures/diagrams//2triangledeloopnopi.tikz}}
\endpgfgraphicnamed &(TR12)\\

\beginpgfgraphicnamed{diagrams//alphacopyw}
\InputIfFileExists{diagrams//alphacopyw.tikz}{}{\input{./figures/diagrams//alphacopyw.tikz}}
\endpgfgraphicnamed&(TR13) &%
\beginpgfgraphicnamed{diagrams//lambdacopyw}
\InputIfFileExists{diagrams//lambdacopyw.tikz}{}{\input{./figures/diagrams//lambdacopyw.tikz}}
\endpgfgraphicnamed &(TR14)\\

  		\hline
  		\end{array}\]
  	\end{center}
  	
  	\caption{Extended ZX-calculus rules for triangle, where $\lambda  \geq 0, \alpha \in [0,~2\pi).$}\label{figure2}
  \end{figure}

  
  \begin{figure}[!h]
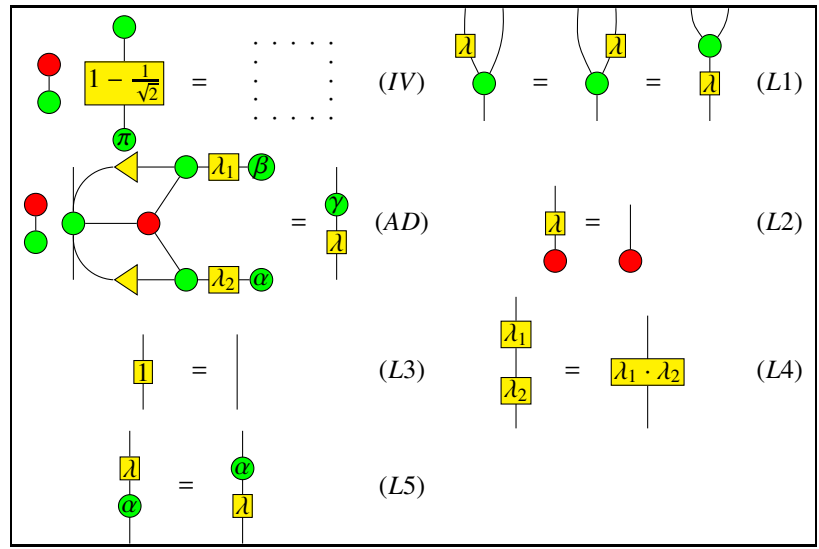

\begin{center}
\[
\quad \qquad\begin{array}{|cccc|}
\hline
\beginpgfgraphicnamed{diagrams//emptyrule}
\InputIfFileExists{diagrams//emptyrule.tikz}{}{\input{./figures/diagrams//emptyrule.tikz}}
\endpgfgraphicnamed &(IV) &%
\beginpgfgraphicnamed{diagrams//lambbranch}
\InputIfFileExists{diagrams//lambbranch.tikz}{}{\input{./figures/diagrams//lambbranch.tikz}}
\endpgfgraphicnamed &(L1)\\

\beginpgfgraphicnamed{diagrams//plus}
\InputIfFileExists{diagrams//plus.tikz}{}{\input{./figures/diagrams//plus.tikz}}
\endpgfgraphicnamed&(AD) &%
\beginpgfgraphicnamed{diagrams//lambdadelete}
\InputIfFileExists{diagrams//lambdadelete.tikz}{}{\input{./figures/diagrams//lambdadelete.tikz}}
\endpgfgraphicnamed &(L2)\\

\beginpgfgraphicnamed{diagrams//sqr1is1}
\InputIfFileExists{diagrams//sqr1is1.tikz}{}{\input{./figures/diagrams//sqr1is1.tikz}}
\endpgfgraphicnamed&(L3) &%
\beginpgfgraphicnamed{diagrams//lambdatimes}
\InputIfFileExists{diagrams//lambdatimes.tikz}{}{\input{./figures/diagrams//lambdatimes.tikz}}
\endpgfgraphicnamed&(L4)\\

\beginpgfgraphicnamed{diagrams//lambdaalpha}
\InputIfFileExists{diagrams//lambdaalpha.tikz}{}{\input{./figures/diagrams//lambdaalpha.tikz}}
\endpgfgraphicnamed&(L5) &&\\
\hline
\end{array}\]
\end{center}
  \caption{Extended ZX-calculus rules for $\lambda$ and addition, where $\lambda, \lambda_1,  \lambda_2 \geq 0, \alpha, \beta, \gamma \in [0,~2\pi);$ in (AD), $\lambda e^{i\gamma} 
  =\lambda_1 e^{i\beta}+ \lambda_2 e^{i\alpha}$.}\label{figure0}  
  \end{figure}


\FloatBarrier
Note that the upside-down versions  of all the above listed rules still hold, thus will be used without being clearly stated.


Now we show that the addition rule (AD) in Figure \ref{figure0} can be simplified.
First, by the symmetry of %
\beginpgfgraphicnamed{diagrams//wstate}
\InputIfFileExists{diagrams//wstate.tikz}{}{\input{./figures/diagrams//wstate.tikz}}
\endpgfgraphicnamed and the rule (TR10), we have
\begin{equation}
\beginpgfgraphicnamed{diagrams//sumsimplify}
\InputIfFileExists{diagrams//sumsimplify.tikz}{}{\input{./figures/diagrams//sumsimplify.tikz}}
\endpgfgraphicnamed
\end{equation}
Therefore,
\begin{equation}
\beginpgfgraphicnamed{diagrams//sumsimplify2}
\InputIfFileExists{diagrams//sumsimplify2.tikz}{}{\input{./figures/diagrams//sumsimplify2.tikz}}
\endpgfgraphicnamed
\end{equation}

As a consequence, we have the following commutativity of addition:
\begin{equation}
\beginpgfgraphicnamed{diagrams//sumcommutativity2}
\InputIfFileExists{diagrams//sumcommutativity2.tikz}{}{\input{./figures/diagrams//sumcommutativity2.tikz}}
\endpgfgraphicnamed
\end{equation}

Next we prove that some rules in Figure \ref{figure2} are derivable. 
\begin{lemma}\label{redundantrules}
The rules (TR4), (TR10), and (TR11) can be derived from other rules.
\end{lemma}

\begin{proof}
For the derivation of (TR4), we have

\begin{equation}
\input{diagrams/TR4derive.tikz}
\end{equation}

where we used (TR3) for the second equality, (TR6) for the third equality, and (TR2) for the last equality.

For the derivation of (TR10), we have
\begin{equation}
\beginpgfgraphicnamed{diagrams//TR10derive}
\InputIfFileExists{diagrams//TR10derive.tikz}{}{\input{./figures/diagrams//TR10derive.tikz}}
\endpgfgraphicnamed
\end{equation}
where we used (TR12) for the second equality, (TR3) for the third equality.

For the derivation of (TR11), we have

\beginpgfgraphicnamed{diagrams//TR11derive}
\InputIfFileExists{diagrams//TR11derive.tikz}{}{\input{./figures/diagrams//TR11derive.tikz}}
\endpgfgraphicnamed

where we used (TR12) for the fourth equality and (TR1) several times.
\end{proof}

If we add a new rule (TR$10^{\prime}$) as shown in Figure \ref{figure2t}, then the rule TR (5) is also derivable.  In fact, 

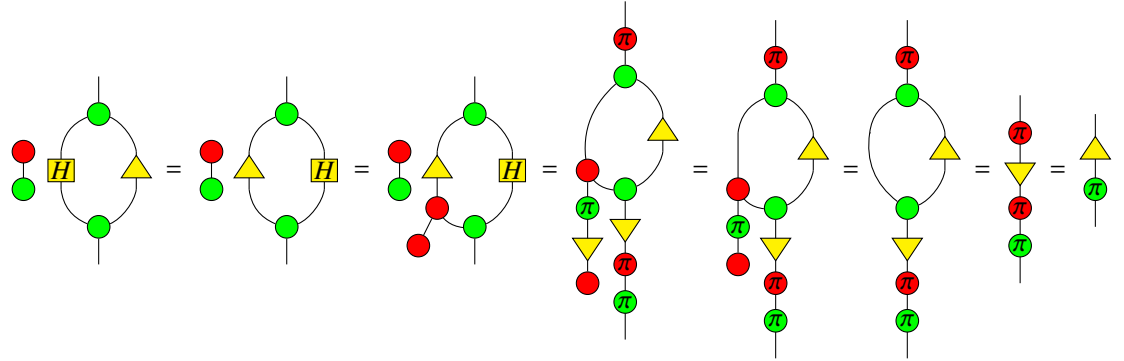
\beginpgfgraphicnamed{diagrams//TR5derive}
\InputIfFileExists{diagrams//TR5derive.tikz}{}{\input{./figures/diagrams//TR5derive.tikz}}
\endpgfgraphicnamed

where for the third equality we used the following diagrammatic reasoning via rules (TR1), (K2), (TR9) and (TR$10^{\prime}$) :

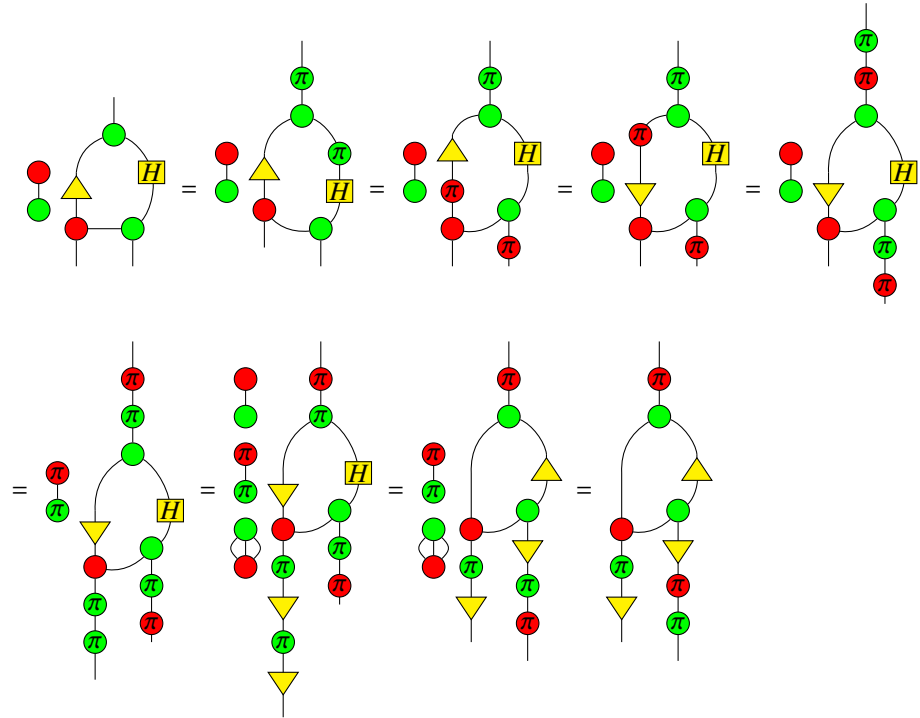
\beginpgfgraphicnamed{diagrams//TR5derive2}
\InputIfFileExists{diagrams//TR5derive2.tikz}{}{\input{./figures/diagrams//TR5derive2.tikz}}
\endpgfgraphicnamed

\section{ZX-calculus for Clifford+T quantum mechanics}

The ZX-calculus for Clifford+T quantum mechanics is a compact closed category $\mathfrak{C}$. The objects of $\mathfrak{C}$ are natural numbers: $0, 1, 2,  \cdots$; the tensor of objects is just addition of numbers: $m \otimes n = m+n$. The morphisms of $\mathfrak{C}$ are diagrams of the ZX-calculus. A general diagram  $D:k\to l$   with $k$ inputs and $l$ outputs is generated by:
\begin{center} 
\begin{tabular}{|r@{~}r@{~}c@{~}c|r@{~}r@{~}c@{~}c|}
\hline
$R_Z^{(n,m)}$&$:$&$n\to m$ & %
\beginpgfgraphicnamed{diagrams//generator_spider}
\InputIfFileExists{diagrams//generator_spider.tikz}{}{\input{./figures/diagrams//generator_spider.tikz}}
\endpgfgraphicnamed & $A$&$:$&$ 1\to 1$& %
\beginpgfgraphicnamed{diagrams//alphagate}
\InputIfFileExists{diagrams//alphagate.tikz}{}{\input{./figures/diagrams//alphagate.tikz}}
\endpgfgraphicnamed\\
\hline
$H$&$:$&$1\to 1$ &%
\beginpgfgraphicnamed{diagrams//HadaDecomSingleslt}
\InputIfFileExists{diagrams//HadaDecomSingleslt.tikz}{}{\input{./figures/diagrams//HadaDecomSingleslt.tikz}}
\endpgfgraphicnamed
 &  $\sigma$&$:$&$ 2\to 2$& %
\beginpgfgraphicnamed{diagrams//swap}
\InputIfFileExists{diagrams//swap.tikz}{}{\input{./figures/diagrams//swap.tikz}}
\endpgfgraphicnamed\\\hline
   $\mathbb I$&$:$&$1\to 1$&%
\beginpgfgraphicnamed{diagrams//Id}
\InputIfFileExists{diagrams//Id.tikz}{}{\input{./figures/diagrams//Id.tikz}}
\endpgfgraphicnamed & $e $&$:$&$0 \to 0$& %
\beginpgfgraphicnamed{diagrams//emptysquare}
\InputIfFileExists{diagrams//emptysquare.tikz}{}{\input{./figures/diagrams//emptysquare.tikz}}
\endpgfgraphicnamed\\\hline
   $C_a$&$:$&$ 0\to 2$& %
\beginpgfgraphicnamed{diagrams//cap}
\InputIfFileExists{diagrams//cap.tikz}{}{\input{./figures/diagrams//cap.tikz}}
\endpgfgraphicnamed &$ C_u$&$:$&$ 2\to 0$&%
\beginpgfgraphicnamed{diagrams//cup}
\InputIfFileExists{diagrams//cup.tikz}{}{\input{./figures/diagrams//cup.tikz}}
\endpgfgraphicnamed \\\hline
\end{tabular}
\end{center}
where $m,n\in \mathbb N$, $\alpha \in \{\frac{k\pi}{4}| k=0, 1, 2, 3, 4, 5, 6, 7 \}$,   and $e$ represents an empty diagram. 

The composition of morphisms is the same as that of the ZX-calculus for overall qubit QM. Due to the angles in the diagrams being multiples of $\frac{\pi}{4}$, we call the ZX-calculus generated by the above generators  
$\frac{\pi}{4}$-fragment ZX-calculus.

\begin{proposition}\label{rings}
The $\frac{\pi}{4}$-fragment of the ZX-calculus exactly corresponds to the matrices over the ring $\mathbb{Z}[i, \frac{1}{\sqrt{2}}]$.
\end{proposition}
\begin{proof}
First note that $\mathbb{Z}[i, \frac{1}{\sqrt{2}}]=\mathbb{Z}[\frac{1}{2}, e^{i\frac{\pi}{4}}]$. It is also clear that each generator of the $\frac{\pi}{4}$-fragment ZX-calculus corresponds to a matrix over the ring $\mathbb{Z}[i, \frac{1}{\sqrt{2}}]$, 
thus each diagram of the $\frac{\pi}{4}$-fragment ZX-calculus must correspond to a matrix over the ring $\mathbb{Z}[i, \frac{1}{\sqrt{2}}]$. Conversely, each matrix over the ring $\mathbb{Z}[i, \frac{1}{\sqrt{2}}]$ can be represented by a normal form 
in the ZW-calculus with phases belong to the same ring $\mathbb{Z}[\frac{1}{2}, e^{i\frac{\pi}{4}}]$ \cite{amar}, hence can be represented by a diagram of the $\frac{\pi}{4}$-fragment ZX-calculus via the translation from ZW to ZX as described in  \cite{ngwang}.
\end{proof}


As was done for the universal ZX-calculus, we extend the language with two new generators-- the triangle and the $\lambda$ box, which will be shown to be representable in the $\frac{\pi}{4}$-fragment ZX-calculus (see 
lemma \ref{lem:lamb_tri_decomposition2}):

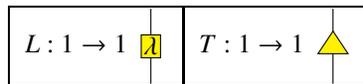
\begin{center} 
	\begin{tabular}{|r@{~}r@{~}c@{~}c|r@{~}r@{~}c@{~}c|}
		\hline
		$L$&$:$&$1\to 1$  &%
\beginpgfgraphicnamed{diagrams//lambdabox}
\InputIfFileExists{diagrams//lambdabox.tikz}{}{\input{./figures/diagrams//lambdabox.tikz}}
\endpgfgraphicnamed &$T$&$:$&$1\to 1$&%
\beginpgfgraphicnamed{diagrams//triangle}
\InputIfFileExists{diagrams//triangle.tikz}{}{\input{./figures/diagrams//triangle.tikz}}
\endpgfgraphicnamed \\\hline
	\end{tabular}
\end{center}
where $0\leqslant \lambda \in \mathbb Z[\frac{1}{2}]$.

There are two kinds of rules for the morphisms of $\mathfrak{C}$:  the structure rules for $\mathfrak{C}$ as an compact closed category, as well as standard rewriting rules listed in Figure \ref{figure1t} and our extended rules listed in Figure \ref{figure2t} and Figure \ref{figure0t}.

Note that all the diagrams should be read from top to bottom.

\begin{figure}[!h]
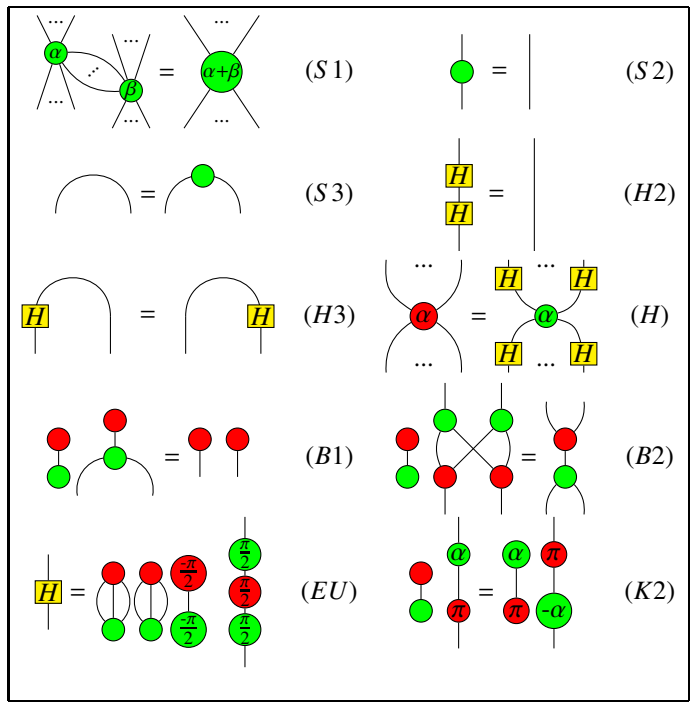

\begin{center}
\[
\quad \qquad\begin{array}{|cccc|}
\hline
\beginpgfgraphicnamed{diagrams//spiderlt}
\InputIfFileExists{diagrams//spiderlt.tikz}{}{\input{./figures/diagrams//spiderlt.tikz}}
\endpgfgraphicnamed=%
\beginpgfgraphicnamed{diagrams//spiderrt}
\InputIfFileExists{diagrams//spiderrt.tikz}{}{\input{./figures/diagrams//spiderrt.tikz}}
\endpgfgraphicnamed &(S1) &%
\beginpgfgraphicnamed{diagrams//s2new}
\InputIfFileExists{diagrams//s2new.tikz}{}{\input{./figures/diagrams//s2new.tikz}}
\endpgfgraphicnamed &(S2)\\
\beginpgfgraphicnamed{diagrams//induced_compact_structure-2wirelt}
\InputIfFileExists{diagrams//induced_compact_structure-2wirelt.tikz}{}{\input{./figures/diagrams//induced_compact_structure-2wirelt.tikz}}
\endpgfgraphicnamed=%
\beginpgfgraphicnamed{diagrams//induced_compact_structure-2wirert}
\InputIfFileExists{diagrams//induced_compact_structure-2wirert.tikz}{}{\input{./figures/diagrams//induced_compact_structure-2wirert.tikz}}
\endpgfgraphicnamed&(S3) & %
\beginpgfgraphicnamed{diagrams//hsquare}
\InputIfFileExists{diagrams//hsquare.tikz}{}{\input{./figures/diagrams//hsquare.tikz}}
\endpgfgraphicnamed &(H2)\\
\beginpgfgraphicnamed{diagrams//hslidecap}
\InputIfFileExists{diagrams//hslidecap.tikz}{}{\input{./figures/diagrams//hslidecap.tikz}}
\endpgfgraphicnamed &(H3) &%
\beginpgfgraphicnamed{diagrams//h2newlt}
\InputIfFileExists{diagrams//h2newlt.tikz}{}{\input{./figures/diagrams//h2newlt.tikz}}
\endpgfgraphicnamed=%
\beginpgfgraphicnamed{diagrams//h2newrt}
\InputIfFileExists{diagrams//h2newrt.tikz}{}{\input{./figures/diagrams//h2newrt.tikz}}
\endpgfgraphicnamed&(H)\\
\beginpgfgraphicnamed{diagrams//b1slt}
\InputIfFileExists{diagrams//b1slt.tikz}{}{\input{./figures/diagrams//b1slt.tikz}}
\endpgfgraphicnamed=%
\beginpgfgraphicnamed{diagrams//b1srt}
\InputIfFileExists{diagrams//b1srt.tikz}{}{\input{./figures/diagrams//b1srt.tikz}}
\endpgfgraphicnamed&(B1) & %
\beginpgfgraphicnamed{diagrams//b2slt}
\InputIfFileExists{diagrams//b2slt.tikz}{}{\input{./figures/diagrams//b2slt.tikz}}
\endpgfgraphicnamed=%
\beginpgfgraphicnamed{diagrams//b2srt}
\InputIfFileExists{diagrams//b2srt.tikz}{}{\input{./figures/diagrams//b2srt.tikz}}
\endpgfgraphicnamed&(B2)\\
\beginpgfgraphicnamed{diagrams//HadaDecomSingleslt}
\InputIfFileExists{diagrams//HadaDecomSingleslt.tikz}{}{\input{./figures/diagrams//HadaDecomSingleslt.tikz}}
\endpgfgraphicnamed= %
\beginpgfgraphicnamed{diagrams//HadaDecomSinglesrt}
\InputIfFileExists{diagrams//HadaDecomSinglesrt.tikz}{}{\input{./figures/diagrams//HadaDecomSinglesrt.tikz}}
\endpgfgraphicnamed&(EU)    & %
\beginpgfgraphicnamed{diagrams//k2slt}
\InputIfFileExists{diagrams//k2slt.tikz}{}{\input{./figures/diagrams//k2slt.tikz}}
\endpgfgraphicnamed=%
\beginpgfgraphicnamed{diagrams//k2srt}
\InputIfFileExists{diagrams//k2srt.tikz}{}{\input{./figures/diagrams//k2srt.tikz}}
\endpgfgraphicnamed&(K2)\\

&&&\\ 
\hline
\end{array}\]
\end{center}
  \caption{Rules for the $\frac{\pi}{4}$-fragment ZX-calculus, where $\alpha, \beta\in  \{\frac{k\pi}{4}| k=0, 1, \cdots, 7\}$.}\label{figure1t}  
  \end{figure}

 \begin{figure}[!h]
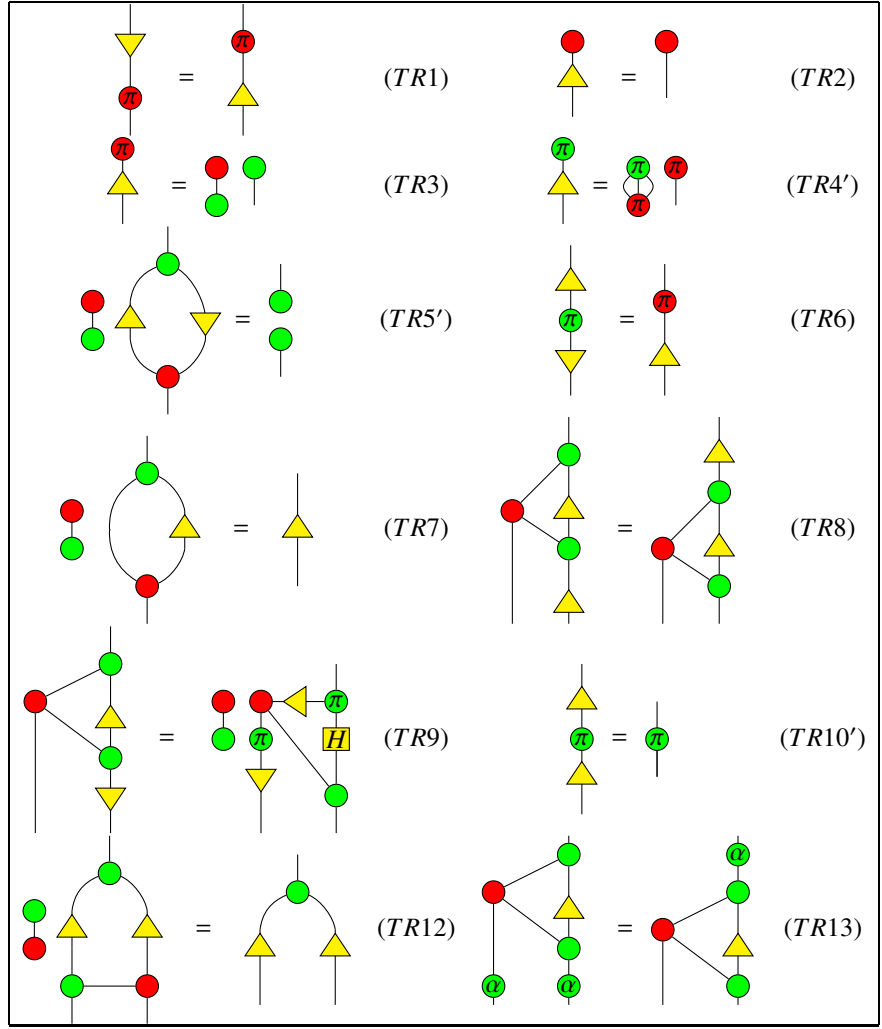

  	\begin{center}
  		\[
  		\quad \qquad\begin{array}{|cccc|}
  		\hline
  		%
\beginpgfgraphicnamed{diagrams//trianglepicommute}
\InputIfFileExists{diagrams//trianglepicommute.tikz}{}{\input{./figures/diagrams//trianglepicommute.tikz}}
\endpgfgraphicnamed &(TR1) &%
\beginpgfgraphicnamed{diagrams//triangleocopy}
\InputIfFileExists{diagrams//triangleocopy.tikz}{}{\input{./figures/diagrams//triangleocopy.tikz}}
\endpgfgraphicnamed &(TR2)\\

\beginpgfgraphicnamed{diagrams//trianglepicopy}
\InputIfFileExists{diagrams//trianglepicopy.tikz}{}{\input{./figures/diagrams//trianglepicopy.tikz}}
\endpgfgraphicnamed&(TR3) & %
\beginpgfgraphicnamed{diagrams//tr4prime}
\InputIfFileExists{diagrams//tr4prime.tikz}{}{\input{./figures/diagrams//tr4prime.tikz}}
\endpgfgraphicnamed &(TR4')\\

\beginpgfgraphicnamed{diagrams//tr5prime}
\InputIfFileExists{diagrams//tr5prime.tikz}{}{\input{./figures/diagrams//tr5prime.tikz}}
\endpgfgraphicnamed &(TR5') &%
\beginpgfgraphicnamed{diagrams//gpiintriangles}
\InputIfFileExists{diagrams//gpiintriangles.tikz}{}{\input{./figures/diagrams//gpiintriangles.tikz}}
\endpgfgraphicnamed&(TR6)\\

\beginpgfgraphicnamed{diagrams//trianglehopf}
\InputIfFileExists{diagrams//trianglehopf.tikz}{}{\input{./figures/diagrams//trianglehopf.tikz}}
\endpgfgraphicnamed&(TR7) & %
\beginpgfgraphicnamed{diagrams//2triangleup}
\InputIfFileExists{diagrams//2triangleup.tikz}{}{\input{./figures/diagrams//2triangleup.tikz}}
\endpgfgraphicnamed&(TR8)\\

\beginpgfgraphicnamed{diagrams//2triangledown}
\InputIfFileExists{diagrams//2triangledown.tikz}{}{\input{./figures/diagrams//2triangledown.tikz}}
\endpgfgraphicnamed&(TR9)    & %
\beginpgfgraphicnamed{diagrams//tr10prime}
\InputIfFileExists{diagrams//tr10prime.tikz}{}{\input{./figures/diagrams//tr10prime.tikz}}
\endpgfgraphicnamed&(TR10')\\

\beginpgfgraphicnamed{diagrams//2triangledeloopnopi}
\InputIfFileExists{diagrams//2triangledeloopnopi.tikz}{}{\input{./figures/diagrams//2triangledeloopnopi.tikz}}
\endpgfgraphicnamed &(TR12) &%
\beginpgfgraphicnamed{diagrams//alphacopyw}
\InputIfFileExists{diagrams//alphacopyw.tikz}{}{\input{./figures/diagrams//alphacopyw.tikz}}
\endpgfgraphicnamed &(TR13)\\
  		
  		
  		\hline
  		\end{array}\]
  	\end{center}
  	
  	\caption{Extended $\frac{\pi}{4}$-fragment ZX-calculus rules for triangle, where $0\leqslant \lambda \in \mathbb Z[\frac{1}{2}], \alpha \in \{\frac{k\pi}{4}| k=0, 1, \cdots, 7\}.$}\label{figure2t}
  \end{figure}

  
  \begin{figure}[!h]
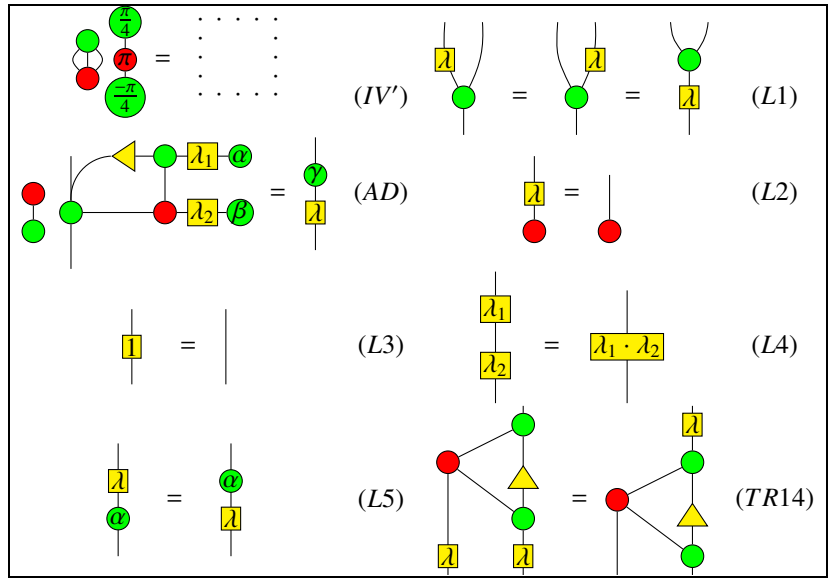

\begin{center}
\[
\quad \qquad\begin{array}{|cccc|}
\hline
%
\beginpgfgraphicnamed{diagrams//newemptyrl}
\InputIfFileExists{diagrams//newemptyrl.tikz}{}{\input{./figures/diagrams//newemptyrl.tikz}}
\endpgfgraphicnamed &(IV') &%
\beginpgfgraphicnamed{diagrams//lambbranch}
\InputIfFileExists{diagrams//lambbranch.tikz}{}{\input{./figures/diagrams//lambbranch.tikz}}
\endpgfgraphicnamed &(L1)\\

%
\beginpgfgraphicnamed{diagrams//plusnew}
\InputIfFileExists{diagrams//plusnew.tikz}{}{\input{./figures/diagrams//plusnew.tikz}}
\endpgfgraphicnamed&(AD) &%
\beginpgfgraphicnamed{diagrams//lambdadelete}
\InputIfFileExists{diagrams//lambdadelete.tikz}{}{\input{./figures/diagrams//lambdadelete.tikz}}
\endpgfgraphicnamed &(L2)\\

\beginpgfgraphicnamed{diagrams//sqr1is1}
\InputIfFileExists{diagrams//sqr1is1.tikz}{}{\input{./figures/diagrams//sqr1is1.tikz}}
\endpgfgraphicnamed&(L3) &%
\beginpgfgraphicnamed{diagrams//lambdatimes}
\InputIfFileExists{diagrams//lambdatimes.tikz}{}{\input{./figures/diagrams//lambdatimes.tikz}}
\endpgfgraphicnamed&(L4)\\

\beginpgfgraphicnamed{diagrams//lambdaalpha}
\InputIfFileExists{diagrams//lambdaalpha.tikz}{}{\input{./figures/diagrams//lambdaalpha.tikz}}
\endpgfgraphicnamed&(L5) &%
\beginpgfgraphicnamed{diagrams//lambdacopyw}
\InputIfFileExists{diagrams//lambdacopyw.tikz}{}{\input{./figures/diagrams//lambdacopyw.tikz}}
\endpgfgraphicnamed &(TR14)\\
\hline
\end{array}\]
\end{center}
  \caption{Extended $\frac{\pi}{4}$-fragment ZX-calculus rules for $\lambda$ and addition, where $0\leqslant \lambda, \lambda_1,  \lambda_2 \in \mathbb Z[\frac{1}{2}], \alpha, \beta, \gamma \in \{\frac{k\pi}{4}| k=0, 1, \cdots, 7\};$ in (AD), $\lambda e^{i\gamma} 
  =\lambda_1 e^{i\alpha}+ \lambda_2 e^{i\beta}$.}\label{figure0t}  
  \end{figure}


\FloatBarrier

Note that the upside-down versions of all the above listed rules still hold.

Since now we focus on the $\frac{\pi}{4}$-fragment ZX-calculus, the empty rule (IV) in Figure \ref{figure0} is changed to the form of rule (IV$^{\prime}$) in Figure \ref{figure0t}. However, we still have the following useful property. 


\begin{lemma}
The frequently used empty rule can be derived from the $\frac{\pi}{4}$-fragment ZX-calculus:
\begin{equation}
\beginpgfgraphicnamed{diagrams//emptyoften}
\InputIfFileExists{diagrams//emptyoften.tikz}{}{\input{./figures/diagrams//emptyoften.tikz}}
\endpgfgraphicnamed
\end{equation}
\end{lemma}

\begin{proof}
\begin{equation*}
\beginpgfgraphicnamed{diagrams//newemptytoold}
\InputIfFileExists{diagrams//newemptytoold.tikz}{}{\input{./figures/diagrams//newemptytoold.tikz}}
\endpgfgraphicnamed
\end{equation*}
\end{proof}

\begin{lemma}\label{lem:lamb_tri_decomposition2}
The triangle %
\beginpgfgraphicnamed{diagrams//triangle}
\InputIfFileExists{diagrams//triangle.tikz}{}{\input{./figures/diagrams//triangle.tikz}}
\endpgfgraphicnamed and the lambda box %
\beginpgfgraphicnamed{diagrams//lambdabox}
\InputIfFileExists{diagrams//lambdabox.tikz}{}{\input{./figures/diagrams//lambdabox.tikz}}
\endpgfgraphicnamed  are expressible in the $\frac{\pi}{4}$-fragment ZX-calculus.
\end{lemma}

\begin{proof}
The triangle %
\beginpgfgraphicnamed{diagrams//triangle}
\InputIfFileExists{diagrams//triangle.tikz}{}{\input{./figures/diagrams//triangle.tikz}}
\endpgfgraphicnamed has been represented in the $\frac{\pi}{4}$-fragment ZX-calculus in \cite{Coeckebk,Emmanuel}, we give the decomposition form according to \cite{Coeckebk} as follows:
\begin{equation}\label{triangleslash}
\beginpgfgraphicnamed{diagrams//triangledecompose}
\InputIfFileExists{diagrams//triangledecompose.tikz}{}{\input{./figures/diagrams//triangledecompose.tikz}}
\endpgfgraphicnamed 
\end{equation}

 Now we deal with the lambda box. 
First we can write $\lambda$ as a sum of its integer part and remainder part: $\lambda= [\lambda] +\{\lambda\}$, where $ [\lambda]$ is a non-negative integer and $0\leq\{\lambda\}<1$. Since  $\lambda \in \mathbb Z[\frac{1}{2}]$,
$\{\lambda\}$ can be uniquely written as a binary expansion of the form $a_1\frac{1}{2}+\cdots+a_s\frac{1}{2^s}$, where $a_i\in \{ 0, 1\} , i=1, \cdots, s.$  For the integer part $[\lambda]$, the corresponding $\lambda$ box 
has been represented in the $\frac{\pi}{4}$-fragment ZX-calculus during the universal completion of the ZX-calculus. For the remainder part $\{\lambda\}$,   it is sufficient to express the $\lambda$ box for
$\lambda=\frac{1}{2}$ in terms of triangle and $Z, X$ phases, since one can apply the addition rule (AD) by recursion.  Actually, we have

\begin{equation}
\beginpgfgraphicnamed{diagrams//lambda1by2}
\InputIfFileExists{diagrams//lambda1by2.tikz}{}{\input{./figures/diagrams//lambda1by2.tikz}}
\endpgfgraphicnamed, \hspace{0.5cm}
\beginpgfgraphicnamed{diagrams//lambda1by2k}
\InputIfFileExists{diagrams//lambda1by2k.tikz}{}{\input{./figures/diagrams//lambda1by2k.tikz}}
\endpgfgraphicnamed 
\end{equation}

  Therefore, 
    $$%
\beginpgfgraphicnamed{diagrams//lexpress4new}
\InputIfFileExists{diagrams//lexpress4new.tikz}{}{\input{./figures/diagrams//lexpress4new.tikz}}
\endpgfgraphicnamed.$$

\end{proof}

The diagrams in the ZX-calculus for Clifford+T QM have the same standard interpretation $\llbracket \cdot \rrbracket$ as that for the whole qubit QM. 

The ZW-calculus for Clifford+T QM almost remain the same as for the whole qubit QM, except that now $r$ in the generator %
\beginpgfgraphicnamed{diagrams//rgatewhite}
\InputIfFileExists{diagrams//rgatewhite.tikz}{}{\input{./figures/diagrams//rgatewhite.tikz}}
\endpgfgraphicnamed  lies in  the ring $\mathbb{Z}[i, \frac{1}{\sqrt{2}}]$.

\section{Features of the new generators }\label{zxforct}
In this section, we show the features of the two generators--the triangle and the  $\lambda$ box. 

 The triangle notation was first introduced in \cite{Emmanuel} as a shortcut for the proof of completeness of the ZX-calculus for Clifford+T QM.  Afterwards, it is employed as a a generator for a complete axiomatisation of the ZX-calculus for the whole pure qubit QM in \cite{ngwang} and for the  Clifford+T fragment in this paper. The purpose to use it as a generator is to make 
 the rewriting rules simple and the translation between the ZX-calculus and the ZW-calculus direct.
 
 Moreover, very recently we find that the triangle can be an essential component for the construction of a Toffoli gate as shown in the following form:
 \begin{equation}\label{toffoligate}
\beginpgfgraphicnamed{diagrams//toffoligtscalar}
\InputIfFileExists{diagrams//toffoligtscalar.tikz}{}{\input{./figures/diagrams//toffoligtscalar.tikz}}
\endpgfgraphicnamed
\end{equation}
where  the triangle with $-1$ on the top-left corner is the inverse of the normal triangle.

In contrast to the standard circuit form which realises the Toffoli gate in elementary gates \cite{Nielsen},  the form of (\ref{toffoligate}) is much more simpler, thus promising for simplifying Clifford + T quantum circuits with the aid of a bunch of  ZX-calculus rules involving triangles.
 
 Unexpectedly, we also realise that the denotation of a slash box used in \cite{Coeckebk} to construct a Toffoli gate is just a triangle (up to a scalar) as shown in  (\ref{triangleslash}).

Next we illustrate the feature of the  $\lambda$ box. In  \cite{ngwang}  and the previous parts of this paper, the  $\lambda$ box is restricted to be parameterised by a non-negative real number. While in  \cite{coeckewang}, it has been generalised to a general green phase of form %
\beginpgfgraphicnamed{diagrams//greenbxa}
\InputIfFileExists{diagrams//greenbxa.tikz}{}{\input{./figures/diagrams//greenbxa.tikz}}
\endpgfgraphicnamed with arbitrary complex number as a parameter. Similarly, we have the general red phase %
\beginpgfgraphicnamed{diagrams//redbxb}
\InputIfFileExists{diagrams//redbxb.tikz}{}{\input{./figures/diagrams//redbxb.tikz}}
\endpgfgraphicnamed  \cite{coeckewang}. Below we give the spider form of general phase which are interpreted in Hilbert spaces:
 \begin{equation}\label{gphaseinter}
 \begin{array}{c}
\left\llbracket %
\beginpgfgraphicnamed{diagrams//generalgreenspider}
\InputIfFileExists{diagrams//generalgreenspider.tikz}{}{\input{./figures/diagrams//generalgreenspider.tikz}}
\endpgfgraphicnamed \right\rrbracket=\ket{0}^{\otimes m}\bra{0}^{\otimes n}+a\ket{1}^{\otimes m}\bra{1}^{\otimes n}, \\
\left\llbracket %
\beginpgfgraphicnamed{diagrams//generalredspider}
\InputIfFileExists{diagrams//generalredspider.tikz}{}{\input{./figures/diagrams//generalredspider.tikz}}
\endpgfgraphicnamed \right\rrbracket=\ket{+}^{\otimes m}\bra{+}^{\otimes n}+a\ket{-}^{\otimes m}\bra{-}^{\otimes n},
\end{array}
\end{equation}
where $a$ is an arbitrary complex number. The generalised spider rules and colour change rule are depicted in the following:
 \begin{equation}\label{generalspider}
 \begin{array}{c}
\beginpgfgraphicnamed{diagrams//generalgreenspiderfuse}
\InputIfFileExists{diagrams//generalgreenspiderfuse.tikz}{}{\input{./figures/diagrams//generalgreenspiderfuse.tikz}}
\endpgfgraphicnamed, \quad   %
\beginpgfgraphicnamed{diagrams//generalredspiderfuse}
\InputIfFileExists{diagrams//generalredspiderfuse.tikz}{}{\input{./figures/diagrams//generalredspiderfuse.tikz}}
\endpgfgraphicnamed,\\
\beginpgfgraphicnamed{diagrams//generalcolorchange}
\InputIfFileExists{diagrams//generalcolorchange.tikz}{}{\input{./figures/diagrams//generalcolorchange.tikz}}
\endpgfgraphicnamed.
\end{array}
\end{equation}
where $a, b$ are arbitrary complex numbers.

Now we consider the generalised supplementarity-- also called cyclotomic supplementarity, with supplementarity as a special case--which is interpreted as merging $n$ subdiagrams if the $n$ phase angles divide the circle uniformly \cite{jpvw}.  We give the diagrammatic form of the generalised supplementarity as follows:
\begin{equation}\label{generalpsgt0}
\beginpgfgraphicnamed{diagrams//generalsupp}
\InputIfFileExists{diagrams//generalsupp.tikz}{}{\input{./figures/diagrams//generalsupp.tikz}}
\endpgfgraphicnamed
\end{equation}
where there are $n$ parallel wires in the diagram at the right-hand side.

Next we show that the generalised supplementarity can be seen as a special form of the generalised spider rule as shown in (\ref{generalspider}). For simplicity, we ignore scalars in the rest of this section.

First note that by comparing the normal form translated from the ZW-calculus \cite{amar}, we  have 
\begin{equation}\label{generalpsgt}
\beginpgfgraphicnamed{diagrams//generalpsgte}
\InputIfFileExists{diagrams//generalpsgte.tikz}{}{\input{./figures/diagrams//generalpsgte.tikz}}
\endpgfgraphicnamed
\end{equation}
where $a\in \mathbb{C}, a\neq 1$. 

Especially,
\begin{equation}\label{generalpsgt2}
\beginpgfgraphicnamed{diagrams//generalpsgte2}
\InputIfFileExists{diagrams//generalpsgte2.tikz}{}{\input{./figures/diagrams//generalpsgte2.tikz}}
\endpgfgraphicnamed
\end{equation}
 where $\alpha \in [0, 2\pi), \alpha\neq \pi$. For $\alpha= \pi$, we can use the $\pi$ copy rule directly. 
  
  Then
  \begin{equation}\label{generalpsgt3}
\beginpgfgraphicnamed{diagrams//generalpsgte3}
\InputIfFileExists{diagrams//generalpsgte3.tikz}{}{\input{./figures/diagrams//generalpsgte3.tikz}}
\endpgfgraphicnamed
\end{equation}
  Note that if $n$ is odd, then
  \begin{equation}\label{generalpsgt4}
\beginpgfgraphicnamed{diagrams//generalpsgte4}
\InputIfFileExists{diagrams//generalpsgte4.tikz}{}{\input{./figures/diagrams//generalpsgte4.tikz}}
\endpgfgraphicnamed
\end{equation}

If $n$ is even, then
 \begin{equation}\label{generalpsgt5}
\beginpgfgraphicnamed{diagrams//generalpsgte5}
\InputIfFileExists{diagrams//generalpsgte5.tikz}{}{\input{./figures/diagrams//generalpsgte5.tikz}}
\endpgfgraphicnamed
\end{equation}

It is not hard to see that if we consider the parity of $n$ in the right diagram of (\ref{generalpsgt0}) with no consideration of scalars, then by Hopf law we get the same result as shown in (\ref{generalpsgt4}) and (\ref{generalpsgt5}).

\section{Interpretations from ZX-calculus to  ZW-calculus and back forth}\label{zwtozx2}
As for the Clifford+T QM,  the interpretation $\llbracket \cdot \rrbracket_{XW}$   from ZX-calculus to  ZW-calculus remains the same:
\[
 \left\llbracket%
\beginpgfgraphicnamed{diagrams//emptysquare}
\InputIfFileExists{diagrams//emptysquare.tikz}{}{\input{./figures/diagrams//emptysquare.tikz}}
\endpgfgraphicnamed\right\rrbracket_{XW}=  %
\beginpgfgraphicnamed{diagrams//emptysquare}
\InputIfFileExists{diagrams//emptysquare.tikz}{}{\input{./figures/diagrams//emptysquare.tikz}}
\endpgfgraphicnamed,  \quad
  \left\llbracket%
\beginpgfgraphicnamed{diagrams//Id}
\InputIfFileExists{diagrams//Id.tikz}{}{\input{./figures/diagrams//Id.tikz}}
\endpgfgraphicnamed\right\rrbracket_{XW}=  %
\beginpgfgraphicnamed{diagrams//Id}
\InputIfFileExists{diagrams//Id.tikz}{}{\input{./figures/diagrams//Id.tikz}}
\endpgfgraphicnamed,   \quad
 \left\llbracket%
\beginpgfgraphicnamed{diagrams//cap}
\InputIfFileExists{diagrams//cap.tikz}{}{\input{./figures/diagrams//cap.tikz}}
\endpgfgraphicnamed\right\rrbracket_{XW}=  %
\beginpgfgraphicnamed{diagrams//cap}
\InputIfFileExists{diagrams//cap.tikz}{}{\input{./figures/diagrams//cap.tikz}}
\endpgfgraphicnamed,   \quad
  \left\llbracket%
\beginpgfgraphicnamed{diagrams//cup}
\InputIfFileExists{diagrams//cup.tikz}{}{\input{./figures/diagrams//cup.tikz}}
\endpgfgraphicnamed\right\rrbracket_{XW}=  %
\beginpgfgraphicnamed{diagrams//cup}
\InputIfFileExists{diagrams//cup.tikz}{}{\input{./figures/diagrams//cup.tikz}}
\endpgfgraphicnamed,  
  \]
  
  \[
   \left\llbracket%
\beginpgfgraphicnamed{diagrams//swap}
\InputIfFileExists{diagrams//swap.tikz}{}{\input{./figures/diagrams//swap.tikz}}
\endpgfgraphicnamed\right\rrbracket_{XW}=  %
\beginpgfgraphicnamed{diagrams//swap}
\InputIfFileExists{diagrams//swap.tikz}{}{\input{./figures/diagrams//swap.tikz}}
\endpgfgraphicnamed,   \quad
   \left\llbracket%
\beginpgfgraphicnamed{diagrams//generator_spider-nonum}
\InputIfFileExists{diagrams//generator_spider-nonum.tikz}{}{\input{./figures/diagrams//generator_spider-nonum.tikz}}
\endpgfgraphicnamed\right\rrbracket_{XW}=  %
\beginpgfgraphicnamed{diagrams//spiderwhite}
\InputIfFileExists{diagrams//spiderwhite.tikz}{}{\input{./figures/diagrams//spiderwhite.tikz}}
\endpgfgraphicnamed,   \quad 
     \left\llbracket%
\beginpgfgraphicnamed{diagrams//alphagate}
\InputIfFileExists{diagrams//alphagate.tikz}{}{\input{./figures/diagrams//alphagate.tikz}}
\endpgfgraphicnamed\right\rrbracket_{XW}=  %
\beginpgfgraphicnamed{diagrams//alphagatewhite}
\InputIfFileExists{diagrams//alphagatewhite.tikz}{}{\input{./figures/diagrams//alphagatewhite.tikz}}
\endpgfgraphicnamed,   \quad 
       \left\llbracket%
\beginpgfgraphicnamed{diagrams//lambdabox}
\InputIfFileExists{diagrams//lambdabox.tikz}{}{\input{./figures/diagrams//lambdabox.tikz}}
\endpgfgraphicnamed\right\rrbracket_{XW}=  %
\beginpgfgraphicnamed{diagrams//lambdagatewhiteld}
\InputIfFileExists{diagrams//lambdagatewhiteld.tikz}{}{\input{./figures/diagrams//lambdagatewhiteld.tikz}}
\endpgfgraphicnamed,   
         \]

 \[
  \left\llbracket%
\beginpgfgraphicnamed{diagrams//HadaDecomSingleslt}
\InputIfFileExists{diagrams//HadaDecomSingleslt.tikz}{}{\input{./figures/diagrams//HadaDecomSingleslt.tikz}}
\endpgfgraphicnamed\right\rrbracket_{XW}=  %
\beginpgfgraphicnamed{diagrams//Hadamardwhite}
\InputIfFileExists{diagrams//Hadamardwhite.tikz}{}{\input{./figures/diagrams//Hadamardwhite.tikz}}
\endpgfgraphicnamed,   \quad
    \left\llbracket%
\beginpgfgraphicnamed{diagrams//triangle}
\InputIfFileExists{diagrams//triangle.tikz}{}{\input{./figures/diagrams//triangle.tikz}}
\endpgfgraphicnamed\right\rrbracket_{XW}=  %
\beginpgfgraphicnamed{diagrams//trianglewhite}
\InputIfFileExists{diagrams//trianglewhite.tikz}{}{\input{./figures/diagrams//trianglewhite.tikz}}
\endpgfgraphicnamed, \]
    
    \[ \llbracket D_1\otimes D_2  \rrbracket_{XW} =  \llbracket D_1  \rrbracket_{XW} \otimes  \llbracket  D_2  \rrbracket_{XW}, \quad 
 \llbracket D_1\circ D_2  \rrbracket_{XW} =  \llbracket D_1  \rrbracket_{XW} \circ  \llbracket  D_2  \rrbracket_{XW},
 \]
where $0\leqslant \lambda \in \mathbb Z[\frac{1}{2}], \alpha \in \{\frac{k\pi}{4}| k=0, 1, \cdots, 7\}$.

\begin{lemma}\label{xtowpreservesemantics2}
Suppose $D$ is an arbitrary diagram in ZX-calculus. Then  $\llbracket \llbracket D \rrbracket_{XW}\rrbracket = \llbracket D \rrbracket$.
\end{lemma}
The proof is easy.

On the other hand, the interpretation $\llbracket \cdot \rrbracket_{WX}$  from ZW-calculus to  ZX-calculus for Clifford+T QM  is almost the same as the case of the overall qubit QM except for the $r$-phase part:

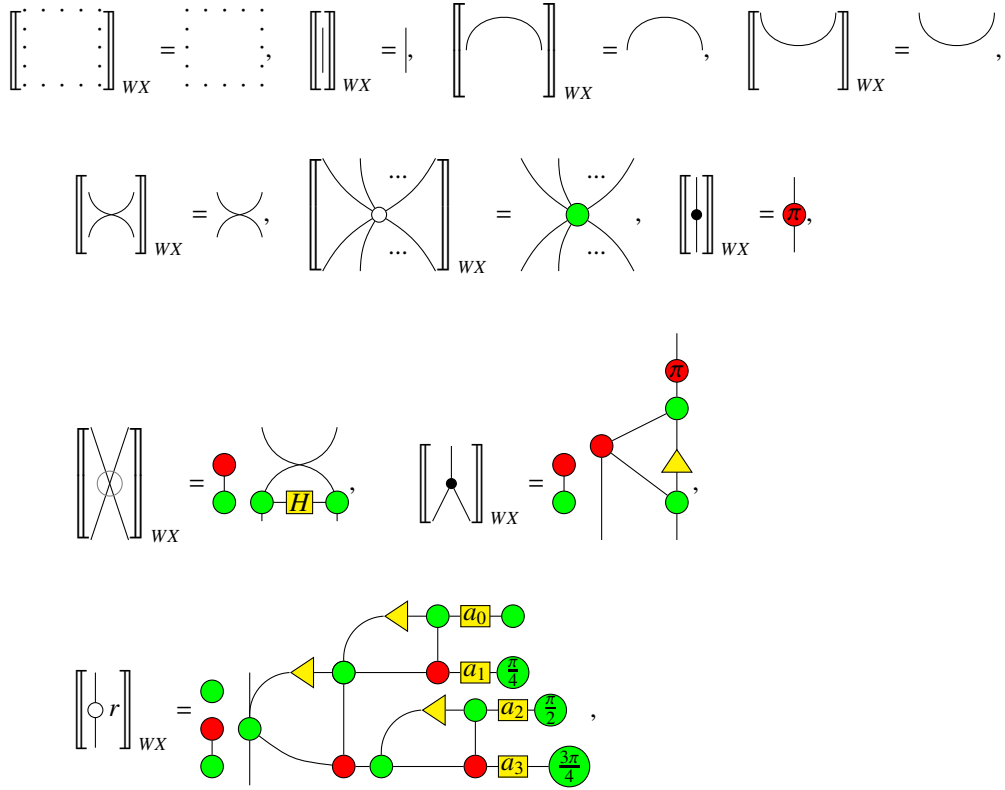
\[
 \left\llbracket%
\beginpgfgraphicnamed{diagrams//emptysquare}
\InputIfFileExists{diagrams//emptysquare.tikz}{}{\input{./figures/diagrams//emptysquare.tikz}}
\endpgfgraphicnamed\right\rrbracket_{WX}=  %
\beginpgfgraphicnamed{diagrams//emptysquare}
\InputIfFileExists{diagrams//emptysquare.tikz}{}{\input{./figures/diagrams//emptysquare.tikz}}
\endpgfgraphicnamed,  \quad
  \left\llbracket%
\beginpgfgraphicnamed{diagrams//Id}
\InputIfFileExists{diagrams//Id.tikz}{}{\input{./figures/diagrams//Id.tikz}}
\endpgfgraphicnamed\right\rrbracket_{WX}=  %
\beginpgfgraphicnamed{diagrams//Id}
\InputIfFileExists{diagrams//Id.tikz}{}{\input{./figures/diagrams//Id.tikz}}
\endpgfgraphicnamed,   \quad
 \left\llbracket%
\beginpgfgraphicnamed{diagrams//cap}
\InputIfFileExists{diagrams//cap.tikz}{}{\input{./figures/diagrams//cap.tikz}}
\endpgfgraphicnamed\right\rrbracket_{WX}=  %
\beginpgfgraphicnamed{diagrams//cap}
\InputIfFileExists{diagrams//cap.tikz}{}{\input{./figures/diagrams//cap.tikz}}
\endpgfgraphicnamed,   \quad
  \left\llbracket%
\beginpgfgraphicnamed{diagrams//cup}
\InputIfFileExists{diagrams//cup.tikz}{}{\input{./figures/diagrams//cup.tikz}}
\endpgfgraphicnamed\right\rrbracket_{WX}=  %
\beginpgfgraphicnamed{diagrams//cup}
\InputIfFileExists{diagrams//cup.tikz}{}{\input{./figures/diagrams//cup.tikz}}
\endpgfgraphicnamed,  
  \]
  
  \[
   \left\llbracket%
\beginpgfgraphicnamed{diagrams//swap}
\InputIfFileExists{diagrams//swap.tikz}{}{\input{./figures/diagrams//swap.tikz}}
\endpgfgraphicnamed\right\rrbracket_{WX}=  %
\beginpgfgraphicnamed{diagrams//swap}
\InputIfFileExists{diagrams//swap.tikz}{}{\input{./figures/diagrams//swap.tikz}}
\endpgfgraphicnamed,   \quad
   \left\llbracket%
\beginpgfgraphicnamed{diagrams//spiderwhite}
\InputIfFileExists{diagrams//spiderwhite.tikz}{}{\input{./figures/diagrams//spiderwhite.tikz}}
\endpgfgraphicnamed\right\rrbracket_{WX}=  %
\beginpgfgraphicnamed{diagrams//generator_spider-nonum}
\InputIfFileExists{diagrams//generator_spider-nonum.tikz}{}{\input{./figures/diagrams//generator_spider-nonum.tikz}}
\endpgfgraphicnamed,   \quad 
       \left\llbracket%
\beginpgfgraphicnamed{diagrams//piblack}
\InputIfFileExists{diagrams//piblack.tikz}{}{\input{./figures/diagrams//piblack.tikz}}
\endpgfgraphicnamed\right\rrbracket_{WX}=  %
\beginpgfgraphicnamed{diagrams//pired}
\InputIfFileExists{diagrams//pired.tikz}{}{\input{./figures/diagrams//pired.tikz}}
\endpgfgraphicnamed,   
         \]

 \[
  \left\llbracket%
\beginpgfgraphicnamed{diagrams//corsszw}
\InputIfFileExists{diagrams//corsszw.tikz}{}{\input{./figures/diagrams//corsszw.tikz}}
\endpgfgraphicnamed\right\rrbracket_{WX}=  %
\beginpgfgraphicnamed{diagrams//crossxz}
\InputIfFileExists{diagrams//crossxz.tikz}{}{\input{./figures/diagrams//crossxz.tikz}}
\endpgfgraphicnamed,   \quad \quad
    \left\llbracket%
\beginpgfgraphicnamed{diagrams//wblack}
\InputIfFileExists{diagrams//wblack.tikz}{}{\input{./figures/diagrams//wblack.tikz}}
\endpgfgraphicnamed\right\rrbracket_{WX}=  %
\beginpgfgraphicnamed{diagrams//winzx}
\InputIfFileExists{diagrams//winzx.tikz}{}{\input{./figures/diagrams//winzx.tikz}}
\endpgfgraphicnamed, \]
    
     \[
 \left\llbracket%
\beginpgfgraphicnamed{diagrams//rgatewhite}
\InputIfFileExists{diagrams//rgatewhite.tikz}{}{\input{./figures/diagrams//rgatewhite.tikz}}
\endpgfgraphicnamed\right\rrbracket_{WX}=%
\beginpgfgraphicnamed{diagrams//rwphasenew}
\InputIfFileExists{diagrams//rwphasenew.tikz}{}{\input{./figures/diagrams//rwphasenew.tikz}}
\endpgfgraphicnamed, 
 \]

    \[ \llbracket D_1\otimes D_2  \rrbracket_{WX} =  \llbracket D_1  \rrbracket_{WX} \otimes  \llbracket  D_2  \rrbracket_{WX}, \quad 
 \llbracket D_1\circ D_2  \rrbracket_{WX} =  \llbracket D_1  \rrbracket_{WX} \circ  \llbracket  D_2  \rrbracket_{WX},
 \]
where $r=a_0+a_1e^{i\frac{\pi}{4}}+a_2e^{i\frac{2\pi}{4}}+a_3e^{i\frac{3\pi}{4}},~   a_j \in \mathbb Z[\frac{1}{2}], ~ j=0, 1, 2, 3$. Note that the representation of $a_j$ box is described in Lemma \ref{lem:lamb_tri_decomposition2}.

\begin{lemma}\label{wtoxpreservesemantics2}
Suppose $D$ is an arbitrary diagram in ZW-calculus. Then  $\llbracket \llbracket D \rrbracket_{WX}\rrbracket = \llbracket D \rrbracket$.
\end{lemma}
The proof is easy.

\begin{lemma}\label{interpretationreversible2}
Suppose $D$ is an arbitrary diagram in ZX-calculus. Then  $ZX\vdash \llbracket \llbracket D \rrbracket_{XW}\rrbracket_{WX} =D$.
\end{lemma}

\begin{proof}
By the construction of  $\llbracket \cdot  \rrbracket_{XW}$  and $\llbracket \cdot  \rrbracket_{WX}$, we only need to prove for $D$ as a generator of ZX-calculus.  The first six generators in ZX-calculus are the same as the  first six generators in ZW-calculus, so we just need to check for the last four generators in ZX-calculus, i.e., the green phase gate, the Hadamard gate, the $\lambda$ box and the triangle.

For the phase gate, we have 
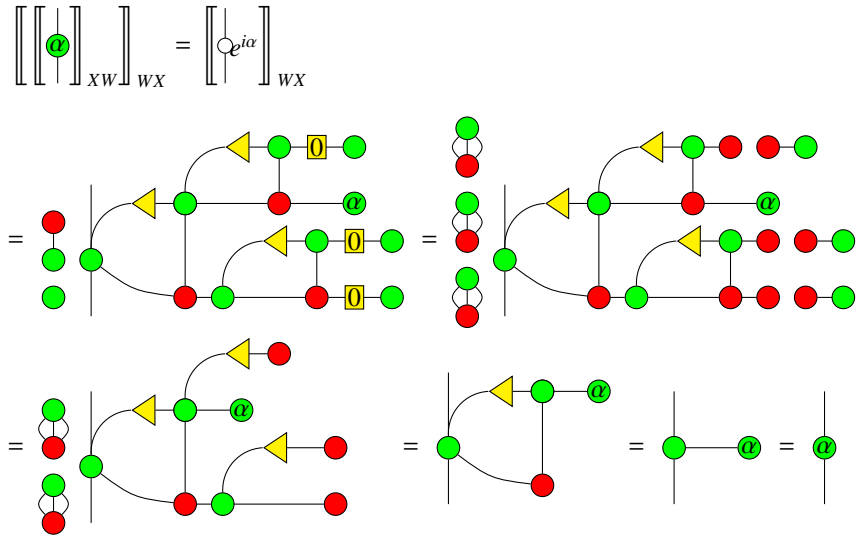
\[
 \left\llbracket \left\llbracket%
\beginpgfgraphicnamed{diagrams//alphagate}
\InputIfFileExists{diagrams//alphagate.tikz}{}{\input{./figures/diagrams//alphagate.tikz}}
\endpgfgraphicnamed\right\rrbracket_{XW}\right\rrbracket_{WX}= \left\llbracket%
\beginpgfgraphicnamed{diagrams//alphagatewhite}
\InputIfFileExists{diagrams//alphagatewhite.tikz}{}{\input{./figures/diagrams//alphagatewhite.tikz}}
\endpgfgraphicnamed\right\rrbracket_{WX}\]
 \[
\beginpgfgraphicnamed{diagrams//alphadrvctnew}
\InputIfFileExists{diagrams//alphadrvctnew.tikz}{}{\input{./figures/diagrams//alphadrvctnew.tikz}}
\endpgfgraphicnamed
\]
For the Hadamard gate, we have

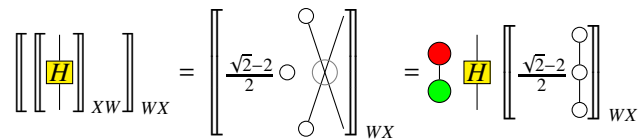
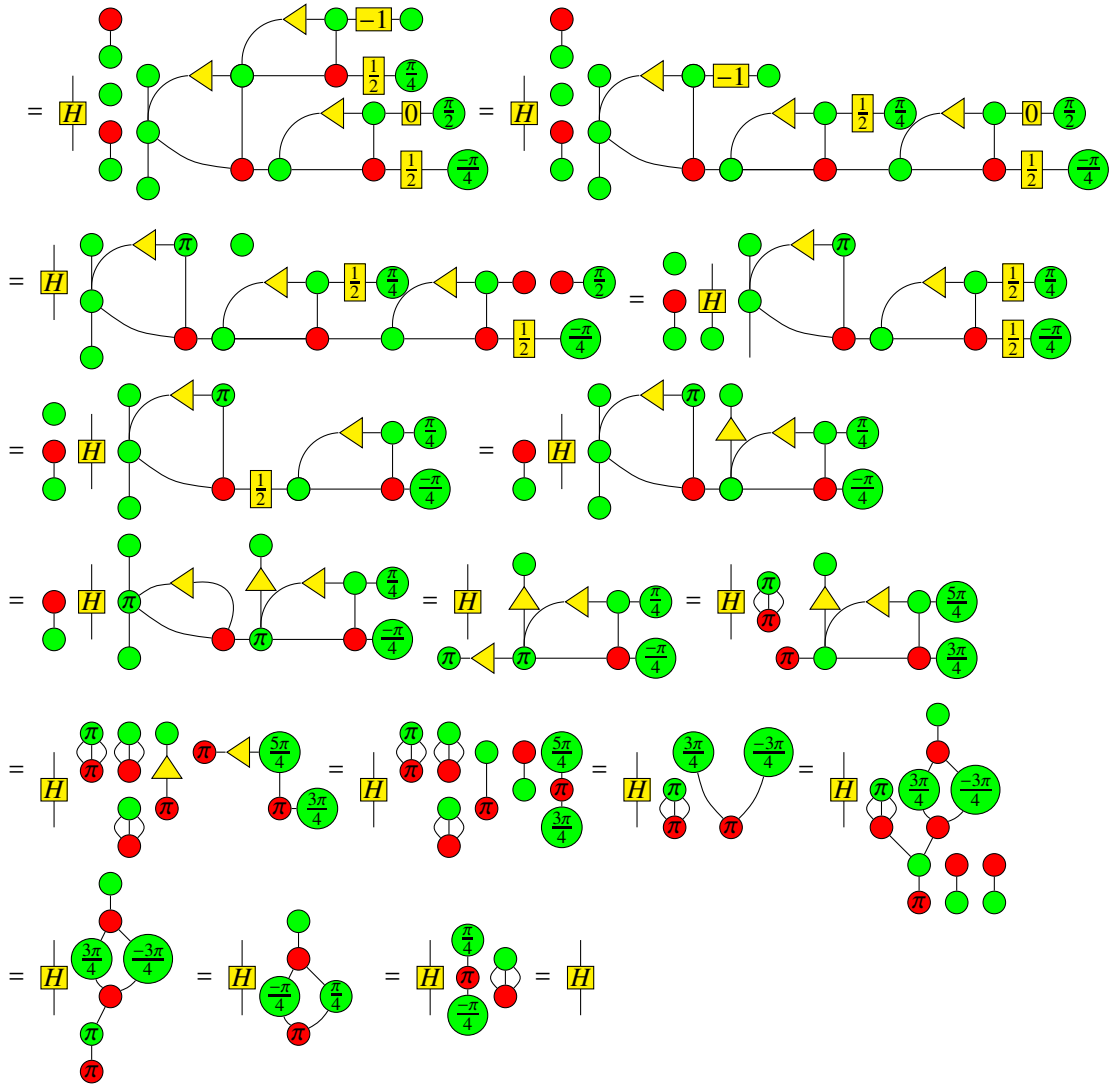
\[
  \left\llbracket \left\llbracket%
\beginpgfgraphicnamed{diagrams//HadaDecomSingleslt}
\InputIfFileExists{diagrams//HadaDecomSingleslt.tikz}{}{\input{./figures/diagrams//HadaDecomSingleslt.tikz}}
\endpgfgraphicnamed\right\rrbracket_{XW}\right\rrbracket_{WX}=  
  \left\llbracket  %
\beginpgfgraphicnamed{diagrams//Hadamardwhite}
\InputIfFileExists{diagrams//Hadamardwhite.tikz}{}{\input{./figures/diagrams//Hadamardwhite.tikz}}
\endpgfgraphicnamed  \right\rrbracket_{WX} =    %
\beginpgfgraphicnamed{diagrams//Hadascalar}
\InputIfFileExists{diagrams//Hadascalar.tikz}{}{\input{./figures/diagrams//Hadascalar.tikz}}
\endpgfgraphicnamed \left\llbracket  %
\beginpgfgraphicnamed{diagrams//Hadamardwhitescalar}
\InputIfFileExists{diagrams//Hadamardwhitescalar.tikz}{}{\input{./figures/diagrams//Hadamardwhitescalar.tikz}}
\endpgfgraphicnamed  \right\rrbracket_{WX}\]
  \[
\beginpgfgraphicnamed{diagrams//Hadamardxwx}
\InputIfFileExists{diagrams//Hadamardxwx.tikz}{}{\input{./figures/diagrams//Hadamardxwx.tikz}}
\endpgfgraphicnamed 
\]

Here we used  $\frac{\sqrt{2}-2}{2}=-1+\frac{1}{2}e^{i\frac{\pi}{4}}+0e^{i\frac{2\pi}{4}}-\frac{1}{2}e^{i\frac{3\pi}{4}}=-1+\frac{1}{2}e^{i\frac{\pi}{4}}+0e^{i\frac{2\pi}{4}}+\frac{1}{2}e^{i\frac{-\pi}{4}}.$

Finally, it is easy to check that 
\[
 \left\llbracket \left\llbracket%
\beginpgfgraphicnamed{diagrams//lambdabox}
\InputIfFileExists{diagrams//lambdabox.tikz}{}{\input{./figures/diagrams//lambdabox.tikz}}
\endpgfgraphicnamed\right\rrbracket_{XW}\right\rrbracket_{WX}=%
\beginpgfgraphicnamed{diagrams//lambdabox}
\InputIfFileExists{diagrams//lambdabox.tikz}{}{\input{./figures/diagrams//lambdabox.tikz}}
\endpgfgraphicnamed,\quad 
   \left\llbracket \left\llbracket%
\beginpgfgraphicnamed{diagrams//triangle}
\InputIfFileExists{diagrams//triangle.tikz}{}{\input{./figures/diagrams//triangle.tikz}}
\endpgfgraphicnamed\right\rrbracket_{XW}\right\rrbracket_{WX}=%
\beginpgfgraphicnamed{diagrams//triangle}
\InputIfFileExists{diagrams//triangle.tikz}{}{\input{./figures/diagrams//triangle.tikz}}
\endpgfgraphicnamed.
\]

\end{proof}

\section{Completeness}
\begin{proposition}\label{zwrulesholdinzx2}
If  $ZW\vdash D_1=D_2$, then  $ZX\vdash \llbracket D_1 \rrbracket_{WX} =\llbracket D_2 \rrbracket_{WX}$.

\end{proposition}

\begin{proof}
Since the derivation of equalities in ZW and ZX is made by rewriting rules,  we need only to prove that $ZX \vdash \left\llbracket D_1\right\rrbracket_{WX} = \left\llbracket D_2\right\rrbracket_{WX}$ where  $D_1=D_2$ is a rewriting rule of ZW-calculus. Most proofs of this proposition have been done in the case of universal completion of the ZX-calculus \cite{ngwang}, we only need to check  for the last 5 rules $rng^{r,s}_{\times}$,  $rng^{r,s}_{+}$, $nat^{r}_{c}$, $nat^{r}_{\varepsilon c}$, $ph^{r}$,  which involve white phases in the ZW-calculus for Clifford+T QM. The rules $nat^{r}_{\varepsilon c}$ and $ph^{r}$ are easy to check, we just deal with the rules  $rng^{r,s}_{\times}$,  $rng^{r,s}_{+}$ and $nat^{r}_{c}$ in the appendix.
\end{proof}

\begin{theorem}\label{maintheorem2}
The ZX-calculus is complete for Clifford+T QM:
if $\llbracket D_1 \rrbracket =\llbracket D_2 \rrbracket$, then $ZX\vdash D_1=D_2$.

\end{theorem}

\begin{proof}
Suppose $D_1,  D_2 \in ZX$ and  $\llbracket D_1 \rrbracket =\llbracket D_2 \rrbracket$. Then by lemma \ref{xtowpreservesemantics2},  $\llbracket \llbracket D_1 \rrbracket_{XW}\rrbracket = \llbracket D_1 \rrbracket= \llbracket D_2 \rrbracket=\llbracket \llbracket D_2 \rrbracket_{XW}\rrbracket $.  Thus by the completeness of ZW-calculus in any commutative ring \cite{amar},  $ZW\vdash \llbracket D_2 \rrbracket_{XW}=  \llbracket D_2 \rrbracket_{XW}$.  Now by proposition \ref{zwrulesholdinzx2},  $ZX\vdash \llbracket \llbracket D_1 \rrbracket_{XW}\rrbracket_{WX} =\llbracket \llbracket D_2 \rrbracket_{XW}\rrbracket_{WX}$.
Finally, by lemma \ref{interpretationreversible2},  $ZX\vdash D_1=D_2$.
\end{proof}


\section{Conclusion and further work}

In this paper, we give  a complete axiomatisation of the ZX-calculus for the Clifford+T QM based on our complete axiomatisation for the overall pure qubit QM \cite{ngwang} and the completeness theorem of the ZW-calculus \cite{amar}. We also show the features of our new generators in contrast to the complete axiomatisation for the Clifford+T QM shown in \cite{Emmanuel}.

A natural thing to do next would be applying the rules of this paper to the simplification of Clifford+T quantum circuits.

It is also interesting to incorporate the rules obtained here in the automated graph rewriting system Quantomatic  \cite{Quanto}. 

\section*{Acknowledgement}
 The authors would like to thank  Bob Coecke for the fruitful discussions and invaluable comments.



\section*{Appendix}

\begin{proposition}(ZW rule $nat^{r}_{c}$)~\newline \\
	\begin{equation}\label{wpscopy}
	ZX\vdash
	\left\llbracket~
	\input{diagrams/natrc_LHS.tikz}
	~\right\rrbracket_{WX}
	=
	\left\llbracket~
	\input{diagrams/natrc_RHS.tikz}
	~\right\rrbracket_{WX},
	\end{equation}
	 where $r=a_0+a_1e^{i\frac{\pi}{4}}+a_2e^{i\frac{2\pi}{4}}+a_3e^{i\frac{3\pi}{4}},   a_j \in \mathbb Z[\frac{1}{2}], ~ j=0, 1, 2, 3$.
\end{proposition}

\begin{proof}
	Let $c_k=a_ke^{i\frac{k\pi}{4}},  k=0, 1, 2, 3$. Then
	\begin{equation}\label{wphaseint}
\begin{array}{ll}
\left\llbracket~
	\input{diagrams/sinzwphase.tikz}
	~\right\rrbracket_{WX}
	=%
\beginpgfgraphicnamed{diagrams//rwphasenew}
\InputIfFileExists{diagrams//rwphasenew.tikz}{}{\input{./figures/diagrams//rwphasenew.tikz}}
\endpgfgraphicnamed
	\stackrel{nat_w^w}{=}
\left\llbracket~
\input{diagrams/zwphasefuse.tikz}
~\right\rrbracket_{WX}&\vspace{0.5cm}\\
\stackrel{cut_z}{=}
\left\llbracket~
\input{diagrams/zwphasecopy22.tikz}
~\right\rrbracket_{WX}
\stackrel{ba_{zw}}{=}
\left\llbracket~
\input{diagrams/zwphasecopy33.tikz}
~\right\rrbracket_{WX}
\stackrel{cut_z}{=}
\left\llbracket~
\input{diagrams/zwphasecopy44.tikz}
~\right\rrbracket_{WX}
\end{array}
	\end{equation}


Thus
\begin{equation}
\begin{array}{ll}
\left\llbracket~
	\input{diagrams/natrc_RHS.tikz}
	~\right\rrbracket_{WX}
\stackrel{(\ref{wphaseint})}{=}
\left\llbracket~
\input{diagrams/zwphasecopy4.tikz}
~\right\rrbracket_{WX}
\stackrel{ba_{w}}{=}
\left\llbracket~
\input{diagrams/zwphasecopy5.tikz}
~\right\rrbracket_{WX}
&\vspace{0.5cm}\\
\stackrel{cut_w,sym_w^x, TR13, TR14}{=}
\left\llbracket~
\input{diagrams/zwphasecopy6.tikz}
~\right\rrbracket_{WX}
\stackrel{(\ref{wphaseint})}{=}
\left\llbracket~
	\input{diagrams/natrc_LHS.tikz}
	~\right\rrbracket_{WX}&
\end{array}
\end{equation}   
Note that all the ZW rules we applied here have been proved to be true as well under the interpretation $\llbracket \cdot \rrbracket_{WX}$.
\end{proof}

\begin{proposition}(ZW rule $rng^{r,s}_{+}$)~\newline \\
	\begin{equation}\label{zwplus}
	ZX\vdash
	\left\llbracket~
	\input{diagrams/rngrsp_LHS.tikz}
	~\right\rrbracket_{WX}
	=
	\left\llbracket~
	\input{diagrams/rngrsp_RHS.tikz}
	~\right\rrbracket_{WX},
	\end{equation}
 where $r=a_0+a_1e^{i\frac{\pi}{4}}+a_2e^{i\frac{2\pi}{4}}+a_3e^{i\frac{3\pi}{4}},~s=b_0+b_1e^{i\frac{\pi}{4}}+b_2e^{i\frac{2\pi}{4}}+b_3e^{i\frac{3\pi}{4}},   a_j, b_j \in \mathbb Z[\frac{1}{2}], ~ j=0, 1, 2, 3$.
\end{proposition}

\begin{proof}
Let $c_k=a_ke^{i\frac{k\pi}{4}}, d_k=b_ke^{i\frac{k\pi}{4}}, k=0, 1, 2, 3$. Then we have 
\begin{equation}
\begin{array}{ll}
\left\llbracket~
	\input{diagrams/rngrsp_LHS.tikz}
	~\right\rrbracket_{WX}
\stackrel{(\ref{wphaseint})}{=}
\left\llbracket~
	\input{diagrams/zwadd.tikz}
	~\right\rrbracket_{WX}
\stackrel{nat_w^w}{=}
\left\llbracket~
	\input{diagrams/zwadd1.tikz}
	~\right\rrbracket_{WX}&\vspace{0.5cm}\\
\stackrel{nat_w^w}{=}
\left\llbracket~
	\input{diagrams/zwadd2.tikz}
	~\right\rrbracket_{WX}
\stackrel{(AD)}{=}
\left\llbracket~
	\input{diagrams/zwadd3.tikz}
	~\right\rrbracket_{WX}
	\stackrel{(\ref{wphaseint})}{=} 
	\left\llbracket~
	\input{diagrams/rngrsp_RHS.tikz}
	~\right\rrbracket_{WX}&	
\end{array}
\end{equation}
Note that all the ZW rules we applied here have been proved to be true as well under the interpretation $\llbracket \cdot \rrbracket_{WX}$.

\end{proof}

\begin{proposition}\label{zwphasefusionclt}(ZW rule $rng^{r,s}_{\times}$)~\newline \\
	\begin{equation}
	ZX\vdash
	\left\llbracket~
	\input{diagrams/rngrsx_LHS.tikz}
	~\right\rrbracket_{WX}
	=
	\left\llbracket~
	\input{diagrams/rngrsx_RHS.tikz}
	~\right\rrbracket_{WX},
	\end{equation}
	 where $r=a_0+a_1e^{i\frac{\pi}{4}}+a_2e^{i\frac{2\pi}{4}}+a_3e^{i\frac{3\pi}{4}},~s=b_0+b_1e^{i\frac{\pi}{4}}+b_2e^{i\frac{2\pi}{4}}+b_3e^{i\frac{3\pi}{4}},   a_j, b_j \in \mathbb Z[\frac{1}{2}], ~ j=0, 1, 2, 3$.
\end{proposition}

\begin{proof}
Let $c_k=a_ke^{i\frac{k\pi}{4}}, d_k=b_ke^{i\frac{k\pi}{4}}, k=0, 1, 2, 3$. Then by (\ref{wphaseint}) we have 
\begin{equation}
\left\llbracket~
	\input{diagrams/sinzwphase.tikz}
	~\right\rrbracket_{WX}
	=
\left\llbracket~
\input{diagrams/zwphasecopy44.tikz}
~\right\rrbracket_{WX},~~
\left\llbracket~
	\input{diagrams/sinzwphase2.tikz}
	~\right\rrbracket_{WX}
	=
\left\llbracket~
\input{diagrams/zwphasefuse22.tikz}
~\right\rrbracket_{WX}
	\end{equation}   
Therefore,
\begin{equation}
\begin{array}{ll}
	\left\llbracket~
	\input{diagrams/rngrsx_LHS.tikz}
	~\right\rrbracket_{WX}
	\stackrel{(\ref{wphaseint})}{=} 
	\left\llbracket~
	\input{diagrams/zwphasefuse33.tikz}
	~\right\rrbracket_{WX}
	\stackrel{ba_w}{=}
	\left\llbracket~
	\input{diagrams/zwphasefuse44.tikz}
	~\right\rrbracket_{WX}&\vspace{0.5cm}\\
	\stackrel{TR13,TR14,nat_w^w}{=}
	\left\llbracket~
	\input{diagrams/zwphasefuse55.tikz}
	~\right\rrbracket_{WX}&\vspace{0.5cm}\\
	\stackrel{TR13,TR14,nat_w^w}{=}
	\left\llbracket~
	\input{diagrams/zwphasefuse66.tikz}
	~\right\rrbracket_{WX}&\vspace{0.5cm}\\
	\stackrel{(\ref{zwplus})}{=}
	\left\llbracket~
	\input{diagrams/rngrsx_RHS.tikz}
	~\right\rrbracket_{WX}&
	\end{array}
	\end{equation}
Note that all the ZW rules we applied here have been proved to be true as well under the interpretation $\llbracket \cdot \rrbracket_{WX}$.

\end{proof}

\end{document}

%% file: tikzstyles.tex
\tikzstyle{env}=[copoint,regular polygon rotate=0,minimum width=0.2cm, fill=black]

\tikzstyle{every picture}=[baseline=-0.25em]
\tikzstyle{dotpic}=[scale=0.5]
\tikzstyle{diredges}=[every to/.style={diredge}]
\tikzstyle{dot graph}=[shorten <=-0.1mm,shorten >=-0.1mm,scale=0.6]
\tikzstyle{plot point}=[circle,fill=black,minimum width=2mm,inner sep=0]

\tikzstyle{braceedge}=[decorate,decoration={brace,amplitude=2mm,raise=-1mm}]
\tikzstyle{small braceedge}=[decorate,decoration={brace,amplitude=1mm,raise=-1mm}]
\tikzstyle{left hook arrow}=[left hook-latex]
\tikzstyle{right hook arrow}=[right hook-latex]

\tikzstyle{dtriangle}=[fill=yellow,draw=black,shape=isosceles triangle,shape border rotate=-90,isosceles triangle stretches=true,inner sep=1pt,minimum width=0.4cm,minimum height=3mm]
\tikzstyle{vtriang}=[fill=yellow,draw=black,shape=isosceles triangle,shape border rotate=180,isosceles triangle stretches=true,inner sep=1pt,minimum width=0.4cm,minimum height=3mm]
\tikzstyle{triangle}=[fill=yellow,draw=black,shape=isosceles triangle,shape border rotate=90,isosceles triangle stretches=true,inner sep=1pt,minimum width=0.4cm,minimum height=3mm]
\tikzstyle{H box}=[rectangle,fill=yellow,draw=black,xscale=1.0,yscale=1.0, inner sep=1.pt]
\tikzstyle{gbox}=[rectangle,fill=green,draw=black,xscale=1.0,yscale=1.0, inner sep=1.pt]
\tikzstyle{rbox}=[rectangle,fill=red,draw=black,xscale=1.0,yscale=1.0, inner sep=1.pt]

\tikzstyle{bn}=[circle,fill=black,draw=black,scale=.4]
\tikzstyle{wn}=[circle,fill=white,draw=black,scale=.6]
\tikzstyle{dn}=[circle,fill=none,draw=gray]

\tikzstyle{black dot}=[inner sep=0.7mm,minimum width=0pt,minimum height=0pt,fill=black,draw=black,shape=circle]

\tikzstyle{dot}=[black dot]
\tikzstyle{smalldot}=[inner sep=0.4mm,minimum width=0pt,minimum height=0pt,fill=black,draw=black,shape=circle]
\tikzstyle{white dot}=[dot,fill=white]
\tikzstyle{antipode}=[white dot,inner sep=0.3mm,font=\footnotesize]
\tikzstyle{smallwhitedot}=[smalldot,fill=white]
\tikzstyle{alt white dot}=[white dot,label={[xshift=3.07mm,yshift=-0.05mm,font=\footnotesize]left:$*$}]
\tikzstyle{gray dot}=[dot,fill=gray!40!white]
\tikzstyle{smallgraydot}=[smalldot,fill=gray!40!white]
\tikzstyle{box vertex}=[draw=black,rectangle]
\tikzstyle{small box}=[box vertex,fill=white]
\tikzstyle{whitebg}=[fill=white,inner sep=2pt]
\tikzstyle{graph state vertex}=[sg vertex,fill=black]

\tikzstyle{wide copoint}=[fill=white,draw=black,shape=isosceles triangle,shape border rotate=90,isosceles triangle stretches=true,inner sep=1pt,minimum width=1.5cm,minimum height=5mm]
\tikzstyle{wide point}=[fill=white,draw=black,shape=isosceles triangle,shape border rotate=-90,isosceles triangle stretches=true,inner sep=1pt,minimum width=1.5cm,minimum height=4mm]
\tikzstyle{very wide copoint}=[fill=white,draw=black,shape=isosceles triangle,shape border rotate=-90,isosceles triangle stretches=true,inner sep=1pt,minimum width=2.5cm,minimum height=4mm]
\tikzstyle{very wide empty copoint}=[draw=black,shape=isosceles triangle,shape border rotate=-90,isosceles triangle stretches=true,inner sep=1pt,minimum width=2.5cm,minimum height=4mm]
\tikzstyle{symm}=[ultra thick,shorten <=-1mm,shorten >=-1mm]

\tikzstyle{square box}=[rectangle,fill=white,draw=black,minimum height=5mm,minimum width=5mm,font=\small]
\tikzstyle{square gray box}=[rectangle,fill=gray!30,draw=black,minimum height=6mm,minimum width=6mm]
\tikzstyle{copoint}=[regular polygon,regular polygon sides=3,draw=black,scale=0.75,inner sep=-0.5pt,minimum width=7mm,fill=white]
\tikzstyle{point}=[regular polygon,regular polygon sides=3,draw=black,scale=0.75,inner sep=-0.5pt,minimum width=7mm,fill=white,regular polygon rotate=180]
\tikzstyle{gray point}=[point,fill=gray!40!white]
\tikzstyle{gray copoint}=[copoint,fill=gray!40!white]

\newcommand{\edgearrow}{{\arrow[black]{>}}}
\newcommand{\edgetick}{{\arrow[black,scale=0.7,very thick]{|}}}

\tikzstyle{diredge}=[->]
\tikzstyle{rdiredge}=[<-]
\tikzstyle{medium diredge}=[->]

\tikzstyle{short diredge}=[->]
\tikzstyle{halfedge}=[-)]
\tikzstyle{other halfedge}=[(-]
\tikzstyle{freeedge}=[(-)]
\tikzstyle{white edge}=[line width=5pt,white]
\tikzstyle{tick}=[postaction=decorate,decoration={markings, mark=at position 0.5 with \edgetick}]
\tikzstyle{small map edge}=[|-latex, gray!60!blue, shorten <=0.9mm, shorten >=0.5mm]
\tikzstyle{thick dashed edge}=[very thick,dashed,gray!40]
\tikzstyle{map edge}=[|-latex,very thick, gray!40, shorten <=1mm, shorten >=0.5mm]
\tikzstyle{tickedge}=[postaction=decorate,
  decoration={markings, mark=at position 0.5 with \edgetick}]
\tikzstyle{dirtickedge}=[postaction=decorate,
  decoration={markings, mark=at position 0.5 with \edgetick},
  decoration={markings, mark=at position 0.85 with \edgearrow}]
\tikzstyle{dirdoubletickedge}=[postaction=decorate,
  decoration={markings, mark=at position 0.4 with \edgetick},
  decoration={markings, mark=at position 0.6 with \edgetick},
  decoration={markings, mark=at position 0.85 with \edgearrow}]

\makeatletter
\newcommand{\boxshape}[3]{%
\pgfdeclareshape{#1}{
\inheritsavedanchors[from=rectangle] 
\inheritanchorborder[from=rectangle]
\inheritanchor[from=rectangle]{center}
\inheritanchor[from=rectangle]{north}
\inheritanchor[from=rectangle]{south}
\inheritanchor[from=rectangle]{west}
\inheritanchor[from=rectangle]{east}
\backgroundpath{
\southwest \pgf@xa=\pgf@x \pgf@ya=\pgf@y
\northeast \pgf@xb=\pgf@x \pgf@yb=\pgf@y

\@tempdima=#2
\@tempdimb=#3

\pgfpathmoveto{\pgfpoint{\pgf@xa - 5pt + \@tempdima}{\pgf@ya}}
\pgfpathlineto{\pgfpoint{\pgf@xa - 5pt - \@tempdima}{\pgf@yb}}
\pgfpathlineto{\pgfpoint{\pgf@xb + 5pt + \@tempdimb}{\pgf@yb}}
\pgfpathlineto{\pgfpoint{\pgf@xb + 5pt - \@tempdimb}{\pgf@ya}}
\pgfpathlineto{\pgfpoint{\pgf@xa - 5pt + \@tempdima}{\pgf@ya}}
\pgfpathclose
}
}}

\boxshape{NEbox}{0pt}{8pt}
\boxshape{SEbox}{0pt}{-8pt}
\boxshape{NWbox}{8pt}{0pt}
\boxshape{SWbox}{-8pt}{0pt}
\makeatother

\tikzstyle{map}=[draw,shape=NEbox,inner sep=7pt]
\tikzstyle{mapdag}=[draw,shape=SEbox,inner sep=7pt]
\tikzstyle{maptrans}=[draw,shape=SWbox,inner sep=7pt]
\tikzstyle{mapconj}=[draw,shape=NWbox,inner sep=7pt]

\tikzstyle{probs}=[shape=semicircle,fill=gray!40!white,draw=black,shape border rotate=180,minimum width=1.2cm]

\tikzstyle{arrs}=[-latex,font=\small,auto]
\tikzstyle{arrow plain}=[arrs]
\tikzstyle{arrow dashed}=[dashed,arrs]
\tikzstyle{arrow bold}=[very thick,arrs]
\tikzstyle{arrow hide}=[draw=white!0,-]
\tikzstyle{arrow reverse}=[latex-]
\tikzstyle{cdnode}=[]

\tikzstyle{gn}=[dot,fill=green,minimum width=0.3cm,inner sep=0pt]
\tikzstyle{rn}=[dot,fill=red,inner sep=0pt,minimum width=0.3cm]

\tikzstyle{rc}=[dot,thick,fill=white,draw = red,minimum width=0.3cm,inner sep=0pt]
\tikzstyle{gc}=[dot,thick,fill=white,draw= green,inner sep=0pt,minimum width=0.3cm]
\tikzstyle{bc}=[dot,thick,fill=white,draw= blue,minimum width=0.3cm]

\tikzstyle{label}=[circle,fill=white,minimum width=0.3cm]

\tikzstyle{clocklabel}=[dot,fill=yellow,draw=black,font=\tiny,inner sep=0.75pt]

\tikzstyle{rsn}=[circle split,draw,fill=red,font=\tiny,inner sep=0.75pt]
\tikzstyle{gsn}=[circle split,draw,fill=green,font=\tiny,inner sep=0.75pt]
\tikzstyle{bsn}=[circle split,draw,fill=blue,font=\tiny,inner sep=0.75pt]

\tikzstyle{rsc}=[circle split,thick,draw= red,draw,fill=white,font=\tiny,inner sep=0.75pt]
\tikzstyle{gsc}=[circle split,thick,draw= green,draw,fill=white,font=\tiny,inner sep=0.75pt]
\tikzstyle{bsc}=[circle split,thick,draw= blue,draw,fill=white,font=\tiny,inner sep=0.75pt]

\tikzstyle{cnot}=[fill=white,shape=circle,inner sep=-1.4pt]
\tikzstyle{wire label}=[font=\tiny, auto]

%% file: defs.tex
\newcommand{\bra}[1]{\ensuremath{\left\langle #1 \right|}}
\newcommand{\ket}[1]{\ensuremath{\left|  #1 \right\rangle}}

\tikzstyle{cdiag}=[matrix of math nodes, row sep=3em, column sep=3em, text height=1.5ex, text depth=0.25ex,inner sep=0.5em]
\tikzstyle{arrow above}=[transform canvas={yshift=0.5ex}]
\tikzstyle{arrow below}=[transform canvas={yshift=-0.5ex}]


%% file: diagrams/alphagate.tikz
\begin{tikzpicture}
	\begin{pgfonlayer}{nodelayer}
		\node [style=none] (0) at (0, -0.5) {};
		\node [style=none] (1) at (0, 0.5) {};
		\node [style=gn] (2) at (0, 0) {$\alpha$};
	\end{pgfonlayer}
	\begin{pgfonlayer}{edgelayer}
		\draw (1.center) to (0.center);
	\end{pgfonlayer}
\end{tikzpicture}

%% file: diagrams/HadaDecomSingleslt.tikz
\begin{tikzpicture}
	\begin{pgfonlayer}{nodelayer}
		\node [style=H box] (0) at (-0.75, 0) {$H$};
		\node [style=none] (1) at (-0.75, -0.5) {};
		\node [style=none] (2) at (-0.75, 0.5) {};
	\end{pgfonlayer}
	\begin{pgfonlayer}{edgelayer}
		\draw (2.center) to (0);
		\draw (1.center) to (0);
	\end{pgfonlayer}
\end{tikzpicture}

%% file: diagrams/Id.tikz
\begin{tikzpicture}
	\begin{pgfonlayer}{nodelayer}
		\node [style=none] (1) at (0.5, 0.3) {};
		\node [style=none] (2) at (0.5, -0.3) {};
		\node [style=none] (3) at (0.5, -0.5) {};
		\node [style=none] (4) at (0.5, 0.5) {};
	\end{pgfonlayer}
	\begin{pgfonlayer}{edgelayer}
		\draw (1.center) to (2.center);
	\end{pgfonlayer}
\end{tikzpicture}

%% file: diagrams/cap.tikz
\begin{tikzpicture}
	\begin{pgfonlayer}{nodelayer}
		\node [style=none] (0) at (0, -0) {};
		\node [style=none] (1) at (1, -0) {};
	\end{pgfonlayer}
	\begin{pgfonlayer}{edgelayer}
		\draw [bend left=90, looseness=1.50] (0.center) to (1.center);
	\end{pgfonlayer}
\end{tikzpicture}

%% file: diagrams/cup.tikz
\begin{tikzpicture}
	\begin{pgfonlayer}{nodelayer}
		\node [style=none] (0) at (0, 0.5) {};
		\node [style=none] (1) at (1, 0.5) {};
	\end{pgfonlayer}
	\begin{pgfonlayer}{edgelayer}
		\draw [bend right=90, looseness=1.50] (0.center) to (1.center);
	\end{pgfonlayer}
\end{tikzpicture}

%% file: diagrams/lambdabox.tikz
\begin{tikzpicture}
	\begin{pgfonlayer}{nodelayer}
		\node [style=H box] (0) at (0, 0) {$\lambda$};
		\node [style=none] (1) at (0, -0.5) {};
		\node [style=none] (2) at (0, 0.5) {};
	\end{pgfonlayer}
	\begin{pgfonlayer}{edgelayer}
		\draw (2.center) to (0);
		\draw (1.center) to (0);
	\end{pgfonlayer}
\end{tikzpicture}

%% file: diagrams/triangle.tikz
\begin{tikzpicture}
	\begin{pgfonlayer}{nodelayer}
		\node [style=none] (0) at (0, 0.5) {};
		\node [style=triangle] (1) at (0, 0) {};
		\node [style=none] (2) at (0, -0.5) {};
	\end{pgfonlayer}
	\begin{pgfonlayer}{edgelayer}
		\draw (0.center) to (2.center);
	\end{pgfonlayer}
\end{tikzpicture}

%% file: diagrams/rgatewhite.tikz
\begin{tikzpicture}
	\begin{pgfonlayer}{nodelayer}
		\node [style=wn] (0) at (0, 0) {};
		\node [style=none] (1) at (0, -0.5) {};
		\node [style=none] (2) at (0.25, 0) {$r$};
		\node [style=none] (3) at (0, 0.5) {};
	\end{pgfonlayer}
	\begin{pgfonlayer}{edgelayer}
		\draw (3.center) to (1.center);
	\end{pgfonlayer}
\end{tikzpicture}

%% file: diagrams/greenbxa.tikz
\begin{tikzpicture}
	\begin{pgfonlayer}{nodelayer}
		\node [style=none] (0) at (0, 0.25) {};
		\node [style=gbox] (1) at (0, 0) {$a$};
		\node [style=none] (2) at (0, -0.25) {};
	\end{pgfonlayer}
	\begin{pgfonlayer}{edgelayer}
		\draw (0.center) to (2.center);
	\end{pgfonlayer}
\end{tikzpicture}

%% file: diagrams/alphagatewhite.tikz
\begin{tikzpicture}
	\begin{pgfonlayer}{nodelayer}
		\node [style=wn] (0) at (0.25, 0) {};
		\node [style=none] (1) at (0.25, -0.5) {};
		\node [style=none] (2) at (0.25, 0.5) {};
		\node [style=none] (3) at (0.5, 0) {$e^{i\alpha}$};
	\end{pgfonlayer}
	\begin{pgfonlayer}{edgelayer}
		\draw (2.center) to (1.center);
	\end{pgfonlayer}
\end{tikzpicture}

%% file: diagrams/lambdagatewhiteld.tikz
\begin{tikzpicture}
	\begin{pgfonlayer}{nodelayer}
		\node [style=none] (0) at (0, -0.5) {};
		\node [style=none] (1) at (0.25, 0) {$\lambda$};
		\node [style=wn] (2) at (0, 0) {};
		\node [style=none] (3) at (0, 0.5) {};
	\end{pgfonlayer}
	\begin{pgfonlayer}{edgelayer}
		\draw (3.center) to (0.center);
	\end{pgfonlayer}
\end{tikzpicture}

%% file: diagrams/piblack.tikz
\begin{tikzpicture}
	\begin{pgfonlayer}{nodelayer}
		\node [style=bn] (0) at (0, 0) {};
		\node [style=none] (1) at (0, -0.5) {};
		\node [style=none] (2) at (0, 0.5) {};
	\end{pgfonlayer}
	\begin{pgfonlayer}{edgelayer}
		\draw (2.center) to (1.center);
	\end{pgfonlayer}
\end{tikzpicture}

%% file: diagrams/pired.tikz
\begin{tikzpicture}
	\begin{pgfonlayer}{nodelayer}
		\node [style=none] (0) at (0, -0.5) {};
		\node [style=none] (1) at (0, 0.5) {};
		\node [style=rn] (2) at (0, 0) {$\pi$};
	\end{pgfonlayer}
	\begin{pgfonlayer}{edgelayer}
		\draw (1.center) to (0.center);
	\end{pgfonlayer}
\end{tikzpicture}

%% file: diagrams/wblack.tikz
\begin{tikzpicture}
	\begin{pgfonlayer}{nodelayer}
		\node [style=bn] (0) at (0, 0) {};
		\node [style=none] (1) at (-0.25, -0.5) {};
		\node [style=none] (2) at (0.25, -0.5) {};
		\node [style=none] (3) at (0, 0.5) {};
	\end{pgfonlayer}
	\begin{pgfonlayer}{edgelayer}
		\draw (3.center) to (0);
		\draw (0) to (1.center);
		\draw (0) to (2.center);
	\end{pgfonlayer}
\end{tikzpicture}

%% file: diagrams/Hadamardwhitescalar.tikz
\begin{tikzpicture}
	\begin{pgfonlayer}{nodelayer}
		\node [style=wn] (0) at (0, 0) {};
		\node [style=wn] (1) at (0, 0.5) {};
		\node [style=wn] (2) at (0, -0.5) {};
		\node [style=none] (3) at (-0.5, 0) {$\frac{\sqrt{2}-2}{2}$};
	\end{pgfonlayer}
	\begin{pgfonlayer}{edgelayer}
		\draw (1) to (0);
		\draw (0) to (2);
	\end{pgfonlayer}
\end{tikzpicture}

%% file: diagrams/sinzwphase.tikz
\begin{tikzpicture}
	\begin{pgfonlayer}{nodelayer}
		\node [style=none] (0) at (0, -0.5) {};
		\node [style=wn] (1) at (0, 0) {};
		\node [style=none] (2) at (0.25, 0) {$r$};
		\node [style=none] (3) at (0, 0.5) {};
	\end{pgfonlayer}
	\begin{pgfonlayer}{edgelayer}
		\draw (1) to (3.center);
		\draw (1) to (0.center);
	\end{pgfonlayer}
\end{tikzpicture}

%% file: diagrams/rngrsp_RHS.tikz
\begin{tikzpicture}
	\begin{pgfonlayer}{nodelayer}
		\node [style=none] (0) at (0, -1) {};
		\node [style=wn] (1) at (0, -0) {};
		\node [style=none] (2) at (0.5, -0) {$r+s$};
		\node [style=none] (3) at (0, 1) {};
	\end{pgfonlayer}
	\begin{pgfonlayer}{edgelayer}
		\draw (1) to (3.center);
		\draw (1) to (0.center);
	\end{pgfonlayer}
\end{tikzpicture}

%% file: diagrams/rngrsx_RHS.tikz
\begin{tikzpicture}
	\begin{pgfonlayer}{nodelayer}
		\node [style=none] (0) at (0, 0.75) {};
		\node [style=none] (1) at (0, -0.7499999) {};
		\node [style=none] (2) at (0.5, -0) {$rs$};
		\node [style=wn] (3) at (0, -0) {};
	\end{pgfonlayer}
	\begin{pgfonlayer}{edgelayer}
		\draw (3) to (0.center);
		\draw (3) to (1.center);
	\end{pgfonlayer}
\end{tikzpicture}

%% file: diagrams/sinzwphase2.tikz
\begin{tikzpicture}
	\begin{pgfonlayer}{nodelayer}
		\node [style=wn] (0) at (0, 0) {};
		\node [style=none] (1) at (0.25, 0) {$s$};
		\node [style=none] (2) at (0, 0.5) {};
		\node [style=none] (3) at (0, -0.5) {};
	\end{pgfonlayer}
	\begin{pgfonlayer}{edgelayer}
		\draw (0) to (2.center);
		\draw (0) to (3.center);
	\end{pgfonlayer}
\end{tikzpicture}